\documentclass[aos,preprint]{imsart}

\RequirePackage[OT1]{fontenc}
\RequirePackage{amsthm,amsmath}
\RequirePackage{natbib}
\RequirePackage[colorlinks,citecolor=blue,urlcolor=blue]{hyperref}

\startlocaldefs
\theoremstyle{plain}
\theoremstyle{exampstyle}

\newtheorem{thm}{Theorem}[section]

\newtheorem{example}[thm]{Example}

\newtheorem{theorem}{Theorem}[section]
\newtheorem{rem}[thm]{Remark}
\newtheorem{lemma}[theorem]{Lemma}
\newtheorem{proposition}[theorem]{Proposition}
\endlocaldefs
\newtheorem{innercustomgeneric}{\customgenericname}
\providecommand{\customgenericname}{}
\newcommand{\newcustomtheorem}[2]{%
	\newenvironment{#1}[1]
	{%
		\renewcommand\customgenericname{#2}%
		\renewcommand\theinnercustomgeneric{##1}%
		\innercustomgeneric
	}
	{\endinnercustomgeneric}
}

\newcustomtheorem{customthm}{Theorem}
\newcustomtheorem{customlemma}{Lemma}

\usepackage{amsmath,amssymb,amsfonts, amsthm, amsbsy, mathtools, epsfig}
\usepackage{subfigure}
\usepackage[usenames]{color}
\usepackage[figuresright]{rotating}
\usepackage{boxedminipage}
\usepackage[colorlinks,citecolor=blue,urlcolor=blue]{hyperref}
\usepackage{epstopdf}
\usepackage{natbib}
\usepackage{enumerate}
\usepackage{verbatim}
\usepackage{bbm}
\usepackage{xcolor}
\usepackage{bm}
\usepackage{mathrsfs,color,dsfont}
\usepackage{algorithm}
\usepackage{algorithmic}

\definecolor{DSgray}{cmyk}{0,1,0,0}

\newcommand{\beq}{\begin{equation}}
\newcommand{\eeq}{\end{equation}}
\newcommand{\beas}{\begin{eqnarray*}}
\newcommand{\eeas}{\end{eqnarray*}}
\newcommand{\bea}{\begin{eqnarray}}
\newcommand{\eea}{\end{eqnarray}}
\newcommand{\bei}{\begin{itemize}}
\newcommand{\eei}{\end{itemize}}
\newcommand{\ben}{\begin{enumerate}}
\newcommand{\een}{\end{enumerate}}
\newcommand{\bet}{\begin{theorem}}
\newcommand{\eet}{\end{theorem}}
\newcommand{\bel}{\begin{lemma}}
\newcommand{\eel}{\end{lemma}}
\newcommand{\bep}{\begin{proposition}}
\newcommand{\eep}{\end{proposition}}
\newcommand{\bed}{\begin{definition}}
\newcommand{\eed}{\end{definition}}
\newcommand{\bec}{\begin{corollary}}
\newcommand{\eec}{\end{corollary}}
\newcommand{\bex}{\begin{example}}
\newcommand{\eex}{\end{example}}

\newcommand{\al}{\alpha}

\newcommand{\argmin}{\mathop{\rm arg\min}}

\renewcommand{\vec}{\boldsymbol}
\def\0{\boldsymbol{0}}

\def\x{\boldsymbol{x}}

\def\R{\mathbb{R}}
\def\V{\boldsymbol{V}}
\def\D{\boldsymbol{D}}
\def\A{\boldsymbol{A}}

\def\x{\boldsymbol{x}}

\def\be{\boldsymbol{\beta}}

\def\C{\boldsymbol{C}}
\def\X{\boldsymbol{X}}

\def\SS{\boldsymbol{S}}
\def\thh{\boldsymbol{\theta}}

\def\0{\boldsymbol{0}}
\def\r{\boldsymbol{r}}
\def\U{\boldsymbol{U}}

\def\x{\boldsymbol{x}}

\def\be{\boldsymbol{\beta}}

\def\X{\boldsymbol{X}}
\def\v{\boldsymbol{v}}
\def\y{\boldsymbol{y}}
\def\e{\boldsymbol{e}}

\def\pr{\mathbb{P}} 
\def\ep{\mathbb{E}} 

\let\hat\widehat

\def\u{\boldsymbol{u}}
\def\v{\boldsymbol{v}}

\begin{document}

\begin{frontmatter}
\title{Quantile Regression Under Memory Constraint}
\runtitle{Quantile Regression Under Memory Constraint}

\begin{aug}
\author{\fnms{Xi}  \snm{Chen}\thanksref{m1}\ead[label=e1]{xchen3@stern.nyu.edu}},
\author{\fnms{Weidong}  \snm{Liu}\thanksref{t1,m2}\ead[label=e2]{weidongl@sjtu.edu.cn}},
\author{\fnms{Yichen}  \snm{Zhang}\thanksref{m1}\ead[label=e3]{yzhang@stern.nyu.edu}}
\thankstext{t1}{Research supported by NSFC, Grant No.11431006 and 11690013,  the Program for Professor of Special Appointment (Eastern Scholar) at Shanghai Institutions of Higher Learning,  Youth Talent Support Program, 973 Program (2015CB856004) and a grant from Australian Research Council.}
\runauthor{X. Chen, W. Liu and Y. Zhang}

\affiliation{New York University\thanksmark{m1} and Shanghai Jiao Tong University \thanksmark{m2}}

\address{X. Chen\\
Y. Zhang\\
Information, Operations and Management Sciences\\
Stern School of Business\\
New York University\\
New York, New York 10012\\
USA\\
\printead{e1}\\
\phantom{E-mail:\ }\printead*{e3}}

\address{W. Liu\\
Department of Mathematics\\
Institute of Natural Sciences and Moe-LSC\\
Shanghai Jiao Tong University\\
Shanghai\\
China\\
\printead{e2}}

\end{aug}

\begin{abstract}
This paper studies the inference problem in quantile regression (QR) for a large sample size $n$ but under a limited memory constraint, where the memory can only store a small batch of data of size $m$. A natural method is the na\"ive divide-and-conquer approach, which splits data into batches of size $m$, computes the local QR estimator for each batch, and then aggregates the estimators via averaging. However, this method only works when $n=o(m^2)$ and is computationally expensive. This paper proposes a computationally efficient method, which only requires an initial QR estimator on a small batch of data and then successively refines the estimator via multiple rounds of aggregations.
Theoretically,  as long as $n$ grows polynomially in $m$, we establish the asymptotic normality for the obtained estimator and show that our estimator with only a few rounds of aggregations achieves the same efficiency as the QR estimator computed on all the data. Moreover, our result allows the case that the dimensionality $p$ goes to infinity. The proposed method can also be applied to address the QR problem under distributed computing environment (e.g., in a large-scale sensor network) or for real-time streaming data.
\end{abstract}



\end{frontmatter}

\section{Introduction}

The development of modern technology has enabled data collection of unprecedented size, which leads to large-scale datasets that cannot be fit into memory or are distributed in many machines over limited memory. For example, the memory of a personal computer only has a storage size in GBs while  the dataset on the hard disk could have a much larger size. In addition, in a sensor network, each sensor is designed to collect and store a limited amount of data, and computations are performed via communications and aggregations among sensors (see, e.g., \cite{Wang17sensor}). Other examples include high-speed data streams that are transient and arrive at the processor at a high speed.  In online streaming computation, the memory is usually limited as compared to the length of the data stream \citep{Gama2013,zhangqi2007}. Under  memory constraints in all these scenarios, classical statistical methods, which are developed under the assumption that the memory can fit all the data, are no longer applicable; thus, many estimation and inference methods need to be re-investigated. For example, suppose that there are $n$ samples for some very large $n$, a fundamental question in data analysis is as follows:

	\bigskip
\addtolength{\fboxsep}{5pt}
\begin{boxedminipage}{.9\textwidth}
How to calculate the sample quantiles of $n$  samples when the memory can only store $m$ samples with $n\gg m$?
\end{boxedminipage}
\bigskip

As one of the most popular interview questions from high-tech companies, this problem has attracted much attention from computer scientists over the last decade; see \cite{manku1998, greenwald2004power, zhangqi2007, guha2009} and the references therein.  However, this is mainly a computation problem with a fixed dataset, which does not involve any statistical modeling.


Motivated by this sample quantile calculation problem, we study a more general problem of quantile regression (QR) under memory constraints.
Quantile regression, which models the conditional quantile of the response variable given covariates, finds a wide range of applications to survival analysis (e.g., \cite{Wang:14:Sur,Xu:16:Sur}), health care (e.g., \cite{Sherwood:13,Luo:13:gap}), and economics (e.g., \cite{Belloni:17}).  
In the classical QR model, assume that there are $n$ \emph{i.i.d.} samples $\{(\X_i, Y_i)\}$ from the following model
\begin{equation}\label{eq:linear_model}
  Y_i=\X_i'\be(\tau)+\epsilon_i, \quad \text{for} \quad i=1,\ldots,n,
\end{equation}
where $\X_i'=(1, X_{i1}, \ldots, X_{ip})$ is the random covariate vector  with the dimension $p+1$ drawn from a common population $\X$. The error $\epsilon_i$ is an unobserved random variable satisfying $\mathbb{P}(\epsilon_i\leq 0|\X_i)=\tau$ for some specified $0<\tau<1$ (known as the quantile level). In other words, $\X_i'\be(\tau)$ is the $\tau$-th quantile of $Y_i$ given $\X_i$.
When all the $n$ samples can be fit into memory, one can estimate $\be(\tau)$ via the classical QR estimator \citep{koenker2005quantile},
\begin{equation}\label{eq:QR}
\widehat{\be}_{QR}=\argmin\limits_{\be\in\mathbb{R}^{p+1}}\sum\limits_{i=1}^n\rho_\tau(Y_i-\X_i'\be),
\end{equation}
where  $\rho_{\tau}(x)=x(\tau-I\{x\leq 0\})$ is the  asymmetric absolute deviation function (a.k.a. check function) and $I(\cdot)$ is the indicator function. However, when samples are distributed across many machines  or the sample size $n$ is extremely large and thus the samples cannot be fit into memory,  it is natural to ask the following question:

\bigskip
\addtolength{\fboxsep}{5pt}
\begin{boxedminipage}{.9\textwidth}
How to estimate and conduct inference about $\be(\tau)$ when the memory can only store $m$ samples with $n\gg m$?
\end{boxedminipage}
\bigskip



The divide-and-conquer (DC), as one of the most  important algorithms in computer science,  has been commonly adopted to deal with this kind of big data challenge. We below describe a general DC algorithm for statistical estimation. Specifically, we split the data indices $\{1,2,...,n\}$  into $N$ subsets $\mathcal{H}_{1},...,\mathcal{H}_{N}$ with equal size $m$ and $N=n/m$.   Correspondingly, the entire dataset $\{Y_{i},\X_{i},1\leq i\leq n\}$ is divided into $N$ batches $\mathcal{D}_{1},\ldots,\mathcal{D}_{N}$, where $\mathcal{D}_k=\{Y_i,\X_i,i\in \mathcal{H}_k\}$ for $1 \leq k \leq N$.
By swapping each batch of data $\mathcal{D}_k$ into the memory, one constructs a low dimension statistic $T_{k}=g_k(\mathcal{D}_{k})$  for $\mathcal{D}_k$ with some function $g_k(\cdot)$. Then the estimator $\widehat{\be}$ is obtained by the aggregation of $\{T_{k}\}_{k=1}^N$ (i.e., $\widehat{\be}=G(T_{1},\ldots,T_{N})$ for some aggregation function $G(\cdot)$).
In recent years, this DC framework has been widely adopted in distributed statistical inference (see, e.g., \cite{li2013statistical,chen2014split,battey2015distributed,zhao2016partially,shi2016massive,banerjee2016divide,volgushev2017distributed} and Section \ref{sec:related} for detailed descriptions).

In addition to memory-constrained estimation on a single machine (where the size of the dataset is much larger than memory size), another natural situation for using the DC framework comes from the application of large-scale wireless sensor networks (see, e.g.,
\cite{shrivastava2004medians,greenwald2004power, wain2006, huang2011sampling,Wang17sensor}).
In a sensor network with $N$ sensors,  the data are collected and stored in different sensors. Moreover, due to limited energy carried by sensors, communication cost is one of the main concerns in  data aggregation. The samples are not transferred to the base station or neighboring sensors directly. Instead, each sensor first summarizes the samples into a low dimensional statistic $T_{k}$, which can be transferred with a low communication cost.  Figure \ref{fig:sensor} visualizes a typical sensor network with data flows as a routing tree with the base station as the root.
An internal sensor node in the $i$-th layer receives  statistics $T_{(\cdot)}$  from its children nodes in $(i+1)$-th layer, and then combines received statistics with its own $T_{(\cdot)}$ and sends the resulting statistic to its parent in the $(i-1)$-th layer. The final estimator in the base station (or central node) can be computed by $\widehat{\be}=G(T_{1},\ldots,T_{N})$. 

\begin{figure}[!t]
\centering
  \includegraphics[width=0.4\textwidth]{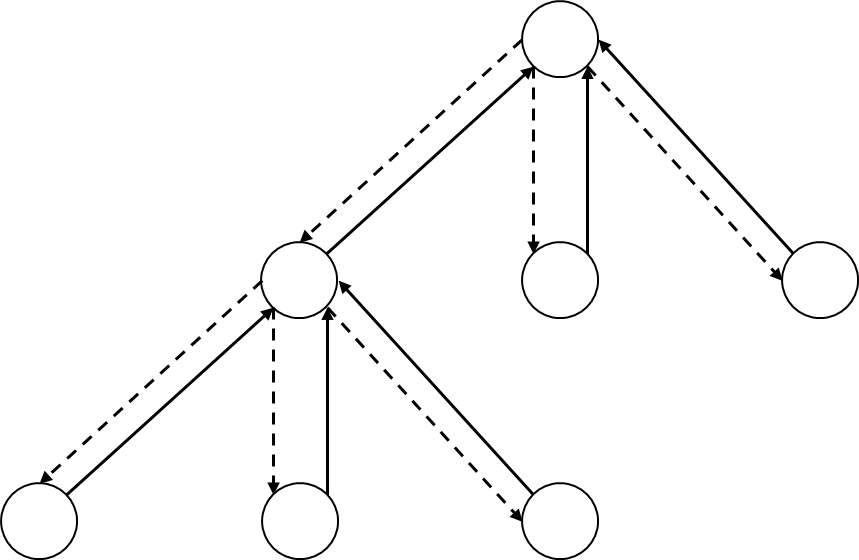}
  \caption{An illustration of tree-structured sensor network where the root is the base station. Each sensor collects its a small batch of data. The dashed lines indicate the prior information sent by the base station to all sensors (e.g., initializations) and the solid lines indicate the information flow (i.e., the paths for transferring local statistics $T_{(\cdot)}$).}
  \label{fig:sensor}
\end{figure}

A critical problem  in statistical DC framework is how to construct local statistics $g_k(\cdot)$ and aggregation function $G(\cdot)$.
In many existing studies, a typical choice of $g_k(\cdot)$ is to use the same estimator as the one designed for the estimation from the entire data. For example, in QR, one may choose
 \begin{eqnarray}\label{dc1}
\widehat{\be}_{QR,k}:=g_k(\mathcal{D}_{k})=\argmin\limits_{\be\in\mathbb{R}^{p+1}}\sum\limits_{i\in\mathcal{H}_{k}}\rho_\tau(Y_i-\X_i'\be),
\end{eqnarray}
and  a simple averaging  function $G(T_{1},\ldots,T_{N})= \frac{1}{N} \sum_{k=1}^N T_k$,  where $T_k=g_k(\mathcal{D}_{k})$ is the local statistic. We call this kind of DC methods the na\"ive-DC algorithm where the estimator is denoted by $\widehat{\be}_{ndc}=\sum_{k=1}^{N}\widehat{\be}_{QR,k}/N$.  Despite its popularity, the na\"ive-DC algorithm might fail when $n/m$ is large. For example, in a special case of quantile estimation (i.e., $p=0$),  it is straightforward to show that $\sqrt{n}|\widehat{\beta}_{ndc}-\beta(\tau)|\rightarrow \infty$ in probability when $n/m^{2}\rightarrow \infty$ (see Theorem \ref{thm:sample_quantile} in Appendix). Similar phenomenon occurs for general $p$, see   \cite{volgushev2017distributed}. In fact, the local QR  estimators $\widehat{\be}_{QR,k}$
are biased estimators with the bias $O(1/m)$. Although the averaging aggregation is useful for variance reduction, it is unable to reduce the bias, which  makes the na\"ive-DC fail when $n$ is large as compared to $m$. In  the DC framework,  bias reduction in $T_{k}$ is more critical than the variance reduction. This is a fundamental difference from many classical inference problems that require to balance the variance and bias (cf. nonparametric estimation). Furthermore, in the na\"ive-DC algorithm, we need to solve $N=n/m$ optimization problems, which could be computationally expensive.

The deficiency of the na\"ive-DC approach calls for a new DC scheme to achieve the following two important goals in distributed inference.
\begin{enumerate}
    \item The obtained estimator $\widehat{\be}$ should achieve the same statistical efficiency as merging all the data together under a weak condition on the sample size $n$ as a function of $m$. More precisely, it is desirable to remove the constraint $n=o(m^2)$ in na\"ive-DC so that the procedure can be applied to situations such as large-scale sensor networks, where the number of sensors $N=n/m$ is excessively large.
    \item The second goal is on the \emph{computational efficiency}. For example, the na\"ive-DC requires solving a non-smooth optimization for computing the local statistic $T_{k}=g_k(\mathcal{D}_k)$. Since there is no explicit formula for $g_k(\cdot)$, the computation is quite heavy (especially considering each sensor has a limited computational power).  
\end{enumerate}

This paper develops new constructions of $g_k(\cdot)$ and $G(\cdot)$ with a multi-round aggregation scheme in  Algorithm \ref{algo:dc-leqr}, which simultaneously  achieves the  two goals referred above. Our method is applicable to both scenarios of small memory on a single machine and large-scale sensor networks.  Instead of using the local QR estimation as in \eqref{dc1}, we adopt a smoothing technique in the literature (see, e.g., \cite{horowitz1998bootstrap, whang2006smoothed, wang2012corrected, pang2012variance, wu2015smoothed}), and propose a new estimator for QR called \emph{linear estimator of QR (LEQR)}, which serves as the cornerstone for our DC approach.  Our LEQR has an explicit formula in the form of direct sums of the transformation of $\{Y_{i}, \X_i\}$, which is  quite different from the optimization-based QR estimator in \eqref{dc1}. It is also worthwhile noting that the linearity is the most desirable property  in the DC framework for both theoretical development and computation efficiency.

The high-level description of the proposed multi-round DC approach is provided as follows. Our method only needs to compute an initial QR estimator $\widehat{\be}_0$  based on a  small part of samples (e.g., $\mathcal{D}_{1}$). Based on $\widehat{\be}_0$, for each batch of data $\mathcal{D}_k$, we compute local statistics $\{T_k\}$ using the proposed LEQR. The local statistics are in a simple form of weighted sums of $\X_i$ and $\X_i \X_i'$. The aggregation function is constructed by adding up the local statistics and then solving a linear system, which gives the first-round estimator $\widehat{\be}^{(1)}$. Now, we can repeat this DC algorithm  using $\widehat{\be}^{(1)}$ as the initial estimator. After $q$ iterations, we denote our final estimator by $\widehat{\be}^{(q)}$.  Theoretically,  under some conditions on the growth rate of $p \rightarrow \infty$ as a function $m$ and $n$, we first establish the Bahadur representation of $\widehat{\be}^{(q)}$ and show that the Bahadur remainder term achieves a nearly optimal rate (up to a log-factor) when $q$ satisfies some mild conditions (see \eqref{eq:q} and  Theorem \ref{thm:beta}). Furthermore, as long as $n=o(m^{A})$ for some constant $A$, the final estimator $\widehat{\be}^{(q)}$ achieves the same asymptotic efficiency as the QR estimator \eqref{eq:QR} computed on the entire data (see Theorem \ref{thm:clt}).


The new DC approach is particularly suitable for QR in sensor networks in which communication cost is one of the major concerns. The proposed procedure  only requires  $O(p^2)$ bits communication between any two sensors. We also highlight  two other important applications of our method:
\begin{enumerate}

\item Our method can be adapted to make inference for online streaming data, which arrives at the processor at a high speed.  Our method provides a sequence of successively refined estimators of $\be(\tau)$ for streaming data and can deal with an arbitrary length of data stream.  The online quantile estimation problem (which is a special case of QR with $p=0$) for streaming data has been extensively studied in computer science literature (see, e.g.,  \cite{munro1980, zhangqi2007, guha2009, wang2003exp}). However, these works mainly focus on developing approximations to the sample quantile, which are insufficient to obtain limiting distribution results for the purpose of inference. We extend the quantile estimation to the more general QR problem and provide the asymptotical normality result for the proposed online  estimator  (see Theorem \ref{th4.3}).

 \item Our method also serves as an efficient optimization solver for  classical QR on a single machine.   As compared to the standard interior-point method for solving the QR estimator in \eqref{eq:QR} that requires the computational complexity of $O(n^{1.25} p^3 \log n)$ \citep{portnoy1997}, our approach requires $O(m^{1.25} p^3\log m + np^2+p^{3})$ since it only solves an optimization on a small batch of data for construction initial estimator. Therefore, our method is computationally more efficient. 


\end{enumerate}
We will illustrate them in Section \ref{sec:est} after we provide the detailed description of the method.

\subsection{Organization and notations}

The rest of the paper is organized as follows. In Section \ref{sec:related}, we review the related literature on recent works on distributed estimation and inference. Section \ref{sec:method} describes the proposed inference procedure for QR under memory constraints. Section \ref{sec:theory} presents the theoretical results.  In Section \ref{sec:simu}, we demonstrate the performance of the proposed inference procedure by simulated experiments, followed by conclusions in Section \ref{sec:con}. The proofs and additional experimental results are provided in Appendix.

In the QR model in \eqref{eq:linear_model}, let $F(\cdot|\vec{x})$ and $f(\cdot|\vec{x})$  denote the CDF and PDF of $\epsilon$ conditioning on $\X=\vec{x}$, respectively, throughout the paper. Then, for any $\vec{x}$, we have $F(0|\vec{x})=\tau$.  \label{page:Omega}For two sequences of real numbers $f(n)$ and $g(n)$, let $f(n)=\Omega (g(n))$ denote that $f$ is bounded below by $g$ (up to constant factor) asymptotically. For a set of random variables $X_n$ and a corresponding set of constants $a_n$, $X_n=O_p(a_n)$ means that $X_n/a_n$ is stochastically bounded and $X_n=o_p(a_n)$ means that $X_n/a_n$ converges to zero in probability as $n$ goes to infinity. For a real number $c$, we will use $\lfloor c \rfloor$ to denote largest integer less than or equal to  $c$.
 Finally, denote the Euclidean norm for a vector $\x\in\R^p$ by $\|\x\|_2$, and denote the spectral norm for a matrix $\X$ by $\|\X\|$.

\section{Related works}
\label{sec:related}

The explosive growth of data presents new challenges for many classical statistical  problems. In recent years, a large body of literature has emerged on studying estimation and inference problems under memory constraints or in distributed environments (please see the references described below as well as other works such as \cite{kleiner2014scalable,wang2013parallelizing,Wang:15:MCMC}).  Examples include but are not limited to  density parameter estimation \citep{li2013statistical},  generalized linear regression with non-convex penalties \citep{chen2014split},
kernel ridge regression \citep{zhang2015divide},  high-dimensional sparse linear regression \citep{lee2017communication}, high-dimensional generalized linear models \citep{battey2015distributed}, semi-parametric partial linear models \citep{zhao2016partially}, QR processes \citep{volgushev2017distributed}, $M$-estimators with cubic rate \citep{shi2016massive}, and some non-standard problems where rates of convergence are slower than $n^{1/2}$ and limit distributions are non-Gaussian \citep{banerjee2016divide}.
All these results rely on averaging, where the global estimator is the average of the local estimators computed on each batch of data. For the averaging estimators to achieve the same asymptotic distribution for inference as pooling the data together, it usually requires the number of batches (i.e., the number of machines) to be $o(m)$ (i.e., $n=o(m^2)$).  However, in some applications such as sensor networks and  streaming data, the number of batches can be large.

To address the challenge, instead of using one-shot aggregation via averaging, recent works by \cite{jordan2016communication} and \cite{wang2016efficient} proposed iterative methods with multiple rounds of aggregations, which relaxes the condition $n=o(m^2)$. These methods have been applied to $M$-estimator and Bayesian inference. Their framework is based on an approximate Newton algorithm \citep{Shamir:14} and thus requires the twice-differentiability of the loss function. However, the QR loss is non-differentiable and thus their approach cannot be utilized. More detailed discussions on these two works will be provided in Remark \ref{rem:comp}. \cite{wain2006} proposed a multi-round decentralized quantile estimation algorithm under a  restrictive communication-constrained setup. However, their method cannot be applied to solve QR problems.

 There is a large body of literature on estimation and inference for QR and its variant (e.g., censored QR). We will not be able to provide a detailed survey here and we refer the readers to \cite{koenker2005quantile,quantilehandbook} for more background knowledge and recent development of QR. However, it is worth noting that the smoothing idea in QR literature has been adopted for developing our linear estimator for QR (LEQR in Section \ref{sec:linear}), which serves as the cornerstone of our method.  The idea of smoothing the non-smooth QR objective goes back to \cite{horowitz1998bootstrap}, where he studied the bootstrap refinement in inference in quantile models. Since the smoothing idea overcomes the difficulty in higher-order expansion of the scores associated with the QR objective, it plays an important role in solving various QR problems. For example, \cite{wang2012corrected} and \cite{wu2015smoothed} proposed different smoothed objectives to determine the corrected scores under the presence of covariate measurement errors for QR and censored QR. \cite{Galvao16smooth} proposed a fix-effects estimator for the smoothed QR in linear panel data models and derived the corresponding limiting distribution.  \cite{pang2012variance} proposed an induced-smoothing idea for estimating the variance of inverse-censoring-probability weighted estimator in \cite{Bang02median}.  \cite{whang2006smoothed} considered the problem of inference using the empirical likelihood mehod for QR and demonstrated that the smoothed empirical likelihood can help achieve higher-order refinements (i.e., $O(n^{-1})$ of the coverage error). The smoothing idea can also facilitate the computation, especially for the first-order optimization methods (see, e.g., \cite{zheng2011gradient}). We also adopt the smoothing idea for constructing our LEQR estimator (see Section \ref{sec:linear}), which heavily relies on the first-order optimality condition of the objective. Instead of using the smoothing technique for computing a one-stage estimator as in existing literature, our use of the smoothing technique enables successive refinement of the LEQR estimator (see Propositions \ref{prop1} and \ref{prop2}).

\section{Methodology}
\label{sec:method}

In this section, we introduce the proposed method. We start with a new linear type estimator for quantile regression, which serves as an important building block for our inference approach.

\subsection{A  linear type  estimator of quantile regression}
\label{sec:linear}
We first propose a linear type estimator for quantile regression, which is named as LEQR (Linear Estimator for Quantile Regression). 
Recall the classical quantile regression estimator from \eqref{eq:QR}. Note that for the quantile regression, the loss function $\rho_{\tau}(x)=x(\tau-I\{x\leq 0\})=x(I\{x> 0\}+\tau-1)$ is non-differentiable. Using the smoothing idea (see the literature surveyed in Section \ref{sec:related}), we approximate the indicator factor $I\left\{x>0\right\}$ with a smooth function $H(x/h)$, where $h \rightarrow 0$ is the bandwidth. With this approximation, we replace $\rho_{\tau}(x)$ in quantile regression by $K_h(x)=x(H(x/h)+\tau-1)$ and define
\begin{eqnarray}
\label{eq:SQR}
\widehat{\be}_h=\argmin_{\be\in\mathbb{R}^{p+1}}\sum\limits_{i=1}^n K_h(Y_i-\X_i'\be).
\end{eqnarray}
Now the right hand side   in (\ref{eq:SQR}) is differentiable and we note that $\frac{\mathrm{d} K_h(x)}{\mathrm{d} x}= H(x/h)+\tau-1+(x/h)H'(x/h)$.  Here  the function $H(u)$ is a smooth approximation of the indicator factor $I\{x>0\}$ satisfying $H(u)=1$ when $u\geq 1$ and $H(u)=0$ when $u\leq -1$ (see more details on the condition of $H(\cdot)$ in Condition (C3) on page \pageref{page:condition}). For example, one may choose $H(u)$ for $u \in[-1,1]$  to be the integral of a smooth kernel function with support on $[-1,1]$, e.g., a biweight (or quartic) kernel (see \eqref{eq:H}). We further note that \eqref{eq:SQR} is a non-convex optimization and thus it is difficult to compute the minimizer $\widehat{\be}_h$. However, it is not a concern since \eqref{eq:SQR} is only introduced for the motivation purpose. The proposed LEQR estimator, which is  explicitly defined in \eqref{eq:beta_0} below, does not require to solve \eqref{eq:SQR}.

Since $H(\cdot)$ is a smooth function, by the first-order optimality condition, the solution $\widehat{\be}_{h}$ in \eqref{eq:SQR} satisfies (see Theorem 2.6 in \cite{Beck:14})
\begin{equation}\label{eq:beta_h_org}
\sum\limits_{i=1}^n \X_i\Big\{H\Big(\dfrac{Y_i-\X'_i\widehat{\be}_h}{h}\Big)+\tau-1+\dfrac{Y_i-\X_i'\widehat{\be}_h}{h}H'\Big(\dfrac{Y_i-\X'_i\widehat{\be}_h}{h}\Big)\Big\}=0.
\end{equation}
From \eqref{eq:beta_h_org} we can express $\widehat{\be}_h$ by
{
\begin{eqnarray}\label{eq:beta_h}
\widehat{\be}_h & = & \left(\sum\limits_{i=1}^n \X_i\X_i' \dfrac1h H'\Big(\dfrac{Y_i-\X'_i\widehat{\be}_h}{h}\Big)\right)^{-1} \times \\
&& \left[\sum\limits_{i=1}^n \X_i\Big\{H\Big(\dfrac{Y_i-\X'_i\widehat{\be}_h}{h}\Big)+\tau-1+\dfrac{Y_i}{h}H'\Big(\dfrac{Y_i-\X'_i\widehat{\be}_h}{h}\Big)\Big\}\right]. \nonumber
\end{eqnarray}
}

However, there is no closed-form expression of $\widehat{\be}_{h}$ from this fixed point equation. 
Instead  of using $\widehat{\be}_{h}$ on the right hand side of (\ref{eq:beta_h}), we replace $\widehat{\be}_{h}$ by a consistent initial estimator $\widehat{\be}_{0}$, which leads to the proposed LEQR:
\begin{eqnarray}\label{eq:beta_0}
\widehat{\be}&=&\left(\sum\limits_{i=1}^n \X_i\X_i' \dfrac1h  H'\Big(\dfrac{Y_i-\X'_i\widehat{\be}_0}{h}\Big)\right)^{-1} \times \\
&& \left[\sum\limits_{i=1}^n \X_i\Big\{H\Big(\dfrac{Y_i-\X'_i\widehat{\be}_0}{h}\Big)+\tau-1+\dfrac{Y_i}{h}H'\Big(\dfrac{Y_i-\X'_i\widehat{\be}_0}{h}\Big)\Big\}\right].\nonumber
\end{eqnarray}



Allowing the choice of initial estimator $\widehat{\be}_{0}$ is crucial for our inference procedure under the memory constraint. While it is difficult to solve the  fixed point equation in \eqref{eq:beta_h} when all data cannot be loaded into memory,
we are still able to compute an initial estimator using only a batch of samples, e.g., $\mathcal{D}_{1}$. Given the initial estimator $\widehat{\be}_{0}$, the LEQR in \eqref{eq:beta_h} only depends  on sums of $\X_i$ and $\X_i\X_i'$, which can be easily implemented via a divide-and-conquer scheme (see Section \ref{sec:est} for details). Comparing to the na\"ive-DC method that needs to solve $N$ optimization problems on each batch of data, LEQR only needs to solve one optimization, and thus is computationally  more efficient.  Moreover, our estimator in \eqref{eq:beta_0} is essentially solving a linear equation system, and we do not need to explicitly compute a matrix inversion. There are a number of efficient methods for solving a linear system numerically, such as conjugate gradient method \citep{hestenes1952methods} and stochastic variance reduced gradient method \citep{johnson2013accelerating}.  In our simulation studies, we use the conjugate gradient method, and due to space limitations, more detailed explanations of this method are provided in Appendix \ref{sec:linear_sys}.

\subsection{Divide-and-conquer LEQR}
\label{sec:est}

Based on LEQR, we now introduce a divide-and-conquer LEQR for estimating $\be(\tau)$.   For each batch of data $\mathcal{D}_k$ for $1\leq k\leq N$, let us define the following quantities.
{\small
\begin{eqnarray}\label{eq:U_k}
\U_{k}&=&\sum_{i\in\mathcal{H}_{k}}\X_i\left\{H\left(\dfrac{Y_i-\X'_i\widehat{\be}_0}{h}\right)+\tau-1+\dfrac{Y_i}{h}H'\left(\dfrac{Y_i-\X'_i\widehat{\be}_0}{h}\right)\right\},\\
\V_{k}&=&\sum_{i\in\mathcal{H}_{k}} \X_i\X_i'\frac{1}{h}H'\left(\dfrac{Y_i-\X'_i\widehat{\be}_0}{h}\right).\nonumber
\end{eqnarray}
\normalfont
}
The inference procedure is presented in Algorithm \ref{algo:dc-leqr}.  Using our theory in Theorem \ref{thm:beta} and \ref{thm:clt}, for the $g$-th iteration, we choose the bandwidth $h=h_g= \max(\sqrt{p/n},(p/m)^{2^{g-2}})$ for $1\leq g\leq q$. 


\begin{algorithm}[!t]
    \caption{{\small Divide-and-conquer(DC) LEQR}}
    \label{algo:dc-leqr}
    \hspace*{\algorithmicindent} \hspace{-0.72cm}  {\textbf{Input:} Data batches $\mathcal{D}_k$ for $k=1,\ldots, N$, the number of iterations/aggregations $q$, quantile level $\tau$, smooth function $H$, a sequence of bandwidths $h_g$ for $g=1,\ldots, q$.  }

    \begin{algorithmic}[1]

     \FOR{$g=1,2 \ldots, q$}
          \IF{$g=1$}
              \STATE Calculate the initial estimator (for the first iteration) based on $\mathcal{D}_{1}$:
                 \begin{eqnarray}\label{ag0}
                        \widehat{\be}_{0}=\argmin\limits_{\be\in\mathbb{R}^{p+1}}\sum\limits_{i\in\mathcal{H}_{1}}\rho_\tau(Y_i-\X_i'\be).
                 \end{eqnarray}
           \ELSE
               \STATE Set the initial estimator to be the estimator from the previous iteration:
               \[
                    \widehat{\be}_0=\widehat{\be}^{(g-1)}
               \]
             \ENDIF
         \FOR{$k=1,\dots, N$}
            \STATE Swap data $\mathcal{D}_k$ into the memory  and compute $(\U_{k},\V_{k})$ according to \eqref{eq:U_k} using the bandwidth $h:=h_g$.
            \STATE   Compute and maintain the sums $(\sum_{j=1}^{k}\U_{j},\sum_{j=1}^{k}\V_{j})$ in the memory (and delete $(\U_{k},\V_{k})$).
         \ENDFOR
         \STATE Compute the estimator $\widehat{\be}^{(g)}$:
            \begin{equation}\label{ag1}
                \widehat{\be}^{(g)}=\Big(\sum_{k=1}^{N}\V_{k}\Big)^{-1}\Big(\sum_{k=1}^{N}\U_{k}\Big).
            \end{equation}
       \ENDFOR

    \end{algorithmic}

   \hspace{-7cm} \textbf{Output:}  The final estimator $\widehat{\be}^{(q)}$.
\end{algorithm}

%


Algorithm \ref{algo:dc-leqr} can not only be used under the memory constraint, but also be applied to distributed setting, to reduce computational cost for classical quantile regression, and to deal with streaming data. We illustrate these important applications as follows.
\begin{enumerate}
\item
{\em Quantile regression in large-scale sensor networks.}  Algorithm \ref{algo:dc-leqr} is directly applicable to distributed sensor network with $N$ sensors, where each sensor collects a batch of data $\mathcal{D}_k$ (for $k=1,\ldots,N$). The base station first broadcasts the initial estimator $\widehat{\be}_{0}$  computed on $\mathcal{D}_1$ in \eqref{ag0} to all sensors (see Figure \ref{fig:sensor} for an illustration). Then each sensor computes $(\U_k, \V_k)$ locally, which will be transferred from bottom to the base station. In particular, each sensor $k$ only keeps the \emph{summation} of $(\U_{\cdot}, \V_{\cdot})$ from all its children nodes  and its own $(\U_k, \V_k)$ and then transfers the summed statistics to its parent.
After receiving $\left(\sum_{k=1}^{N}\U_{k}, \sum_{k=1}^{N}\V_{k}\right)$, the base station will compute $\widehat{\be}^{(g)}$ for $g=1$ in \eqref{ag1}. Then this distributed procedure can be repeated for $g=2,\ldots, q$.

\item
{\em Computational reduction of quantile regression.}   Algorithm \ref{algo:dc-leqr} can also be utilized as an efficient solver for  classical quantile regression  on a single machine.   For the ease of illustration, let us assume the quantile regression estimator in \eqref{eq:QR} (or \eqref{dc1}) is solved by the standard interior-point method.  
Using the standard interior-point method, the initial estimator requires the computational complexity $O(m^{1.25} p^3\log m)$ \citep{portnoy1997}.  The computation of $\sum_{k=1}^N \U_k$ and $\sum_{k=1}^N \V_k$ require $O(np)$ and $O(np^{2})$, respectively.
Therefore, the computational complexity for $\widehat{\be}^{(q)}$ is at most  $O(m^{1.25} p^3\log m + np^2 +p^{3})$, where $O(p^3)$ comes from the inversion of $\sum_{k=1}^N \V_k$.  This greatly saves the computational cost as compared to the interior-point method for  computing quantile regression estimator on the entire data, which requires a  complexity of $O(n^{1.25} p^3\log n)$.

 \begin{algorithm}[!t]
    \caption{{\small Online LEQR}}
    \label{algo:online_LEQR}
\raggedright
  \hspace*{\algorithmicindent} \hspace{-0.72cm}  {\textbf{Initialization:} For the first $m$ samples $\{Y_{i}, \X_{i}, -m+1\leq i\leq 0\}$, compute the initial standard quantile regression estimator   $\widehat{\be}[{0}]$. Then
  constructs the initial $(\U(0), \V(0))$ based on $\widehat{\be}[{0}]$ according to \eqref{eq:U_k} with $h=(p/m)^{1/2}$.

\textbf{Parameter Setup:} Define $a_{2k-1}=2^{k-1}+1/2$ and $a_{2k}=2^{k-1}+3/4$ for $k\geq 1$ and $a_{0}=-\infty$. Let $s_l=\lfloor m^{a_{l-1}}\rfloor+1$ and $r_l=\lfloor  m^{a_{l}}\rfloor$ for any $l\geq 1$ and let $r_0=0$.  Define a sequence of bandwidths $h_{1}=(p/m)^{1/2}$ and $h_{l}=(p/m^{a_{l-1}})^{1/2}$ for $l\geq 2$.
  }

  \begin{algorithmic}[1]
    \FOR{each interval $l=1,2\ldots, $}
        \FOR{indices in the interval $l$, $j=s_l, s_l+1, \ldots, r_l$}
            \STATE Receive an online sample  $(Y_j, \X_j)$.
            \STATE Compute
                \begin{align*}
                 \widetilde{\U}(j) & =\X_j\Bigl\{H\Big(\dfrac{Y_j-\X'_j\widehat{\be}[r_{l-1}])}{h_{l}}\Big)+\tau-1+\dfrac{Y_j}{h_{l}}H'\Big(\dfrac{Y_j-\X'_j\widehat{\be}[r_{l-1}]}{h_{l}}\Big)\Bigr\}
                 \\
                 \widetilde{\V}(j) & =\X_j\X_j'\frac{1}{h_{l}}H'\Big(\dfrac{Y_j-\X'_j\widehat{\be}[r_{l-1}]}{h_{l}}\Big),
                \end{align*}
                where $\widehat{\be}[r_{l-1}]$ is the estimator computed upto the end of $(l-1)$-th interval.
            \STATE Update $(\U(j),\V(j))$ by
    \begin{align*}\label{eq:online_U_j}
    \U(j) & \triangleq \sum_{i=s_l}^j \widetilde{\U}(j) = \begin{cases}
           \widetilde{\U}(j) & \quad \text{if}\; j=s_l\\
           \U(j-1)+ \widetilde{\U}(j)
           & \quad \text{if}\; j> s_l
         \end{cases}\\
     \V(j) & \triangleq \sum_{i=s_l}^j \widetilde{\V}(j) = \begin{cases}
           \widetilde{\V}(j) & \quad \text{if}\; j=s_l\\
           \V(j-1)+ \widetilde{\V}(j)
           & \quad \text{if}\; j> s_l
         \end{cases}
    \end{align*}

            \STATE Calculate $(\U(r_{l-1})+\U(j),  \V(r_{l-1})+\V(j))$ and compute the online quantile regression estimator at time $j$:
             \begin{eqnarray}\label{online}
                 \widehat{\be}[j]=(\V(r_{l-1})+\V(j))^{-1}(\U(r_{l-1})+\U(j)).
             \end{eqnarray}

              \STATE Remove $(\widetilde{\U}(j),\widetilde{\V}(j)), (\U(j-1),\V(j-1))$ from the memory.
        \ENDFOR
        \STATE Remove $(\U(r_{l-1}),\V(r_{l-1}))$ from the memory and only keep  $\widehat{\be}[r_{l}]$, $(\U(r_{l}), \V(r_{l}))$ in the memory.
    \ENDFOR

  \end{algorithmic}

\end{algorithm}

\item {\em Online quantile regression for streaming data.}
To deal with online streaming data, it is critical to design a one-pass algorithm since streaming data are transient. 
To this end, based on Algorithm \ref{algo:dc-leqr}, we develop a new one-pass algorithm (see Algorithm \ref{algo:online_LEQR}) that provides a sequence of successively refined estimators.

For streaming data, we divide the data into intervals $\{(s_l,r_l)\}_{l=1}^\infty$. The starting and ending positions of the $l$-th interval are chosen as
\[
s_l=\lfloor m^{a_{l-1}}\rfloor+1,\text{ and }r_l=\lfloor  m^{a_{l}}\rfloor
\]
for $l\geq 1$ where $a_{2k-1}=2^{k-1}+1/2$ and $a_{2k}=2^{k-1}+3/4$ for $k\geq 1$ and $a_{0}=-\infty$. The intervals are chosen to ensure that the sample size of the $l$-th interval $n_l$ is approximately $n_{l-2}^2$.  As we will show in the proof of Theorem \ref{th4.3}, if an initial estimator is computed from $m$ samples, there will be no improvement of the online LEQR estimator after $m^2$ fresh samples and this is the ending point of an interval where we compute a new initial estimator.

 In  Algorithm \ref{algo:online_LEQR}, for each interval $l$ and each $j$ such that  $s_l \leq j\leq r_l$, the memory only maintains $\widehat{\be}[r_{l-1}]$, $(\U(r_{l-1}), \V(r_{l-1}))$, $(\U(j), \V(j))$, and $\widehat{\be}[j]$. We note that $(\U(r_{l-1}), \V(r_{l-1}))$ are the weighted sums of $\X_i$ and $\X_i\X_i'$ for $s_{l-1}\leq i \leq r_{l-1}$ and $(\U(j), \V(j))$ can be easily updated from $(\U(j-1), \V(j-1))$ in an online fashion. Therefore, except for an $O(m)$ space for deriving the initial estimator $\widehat\be[0]$, the online LEQR only requires $O(p^2)$ memory, which is independent on $n$.

\label{page:online}

In Theorem \ref{th4.3}, we will show that the online LEQR algorithm achieves the same statistical efficiency  for any $l$ and $j$ as the standard quantile regression estimator when merging all the streaming data together. Also,
the asymptotic normality of $\sqrt{m+j}(\widehat{\be}[j]-\be(\tau))$ holds uniformly in $1\leq j\leq m^{A}$ for any constant $A>0$ (note that $m+j$ is the sample size until time $j$).

\end{enumerate}

\section{Theoretical results}
\label{sec:theory}

In this section, we provide a Bahadur representation of $\widehat{\be}^{(q)}$, based on which we derive the  asymptotic normality result for DC LEQR $\widehat{\be}^{(q)}$ and the online LEQR $\widehat{\be}[j]$.  We also discuss adaptive choices of bandwidth and extensions to heterogeneous settings.

\subsection{Asymptotics for DC LEQR}\label{sec:DC_leqr}
We note that \eqref{eq:beta_0} can be equivalently written as:
\begin{equation}\label{eq:key}
\widehat{\be}-\be(\tau)=\D_{n,h}^{-1}\A_{n,h},
\end{equation}
where
\begin{eqnarray}\label{eq:AD}
\A_{n,h}&=&\dfrac{1}{n}\sum\limits_{i=1}^n \X_i\bigg\{H\Big(\dfrac{Y_i-\X'_i\widehat{\be}_0}{h}\Big)+\tau-1+\dfrac{\epsilon_i}{h}H'\Big(\dfrac{Y_i-\X'_i\widehat{\be}_0}{h}\Big)\bigg\},\cr
\D_{n,h}&=&\dfrac{1}{nh}\sum\limits_{i=1}^n \X_i\X_i'H'\Big(\dfrac{Y_i-\X'_i\widehat{\be}_0}{h}\Big) \nonumber.
\end{eqnarray}

We first state some regularity conditions for our theoretical development and then give Propositions \ref{prop1}-\ref{prop2} on the  expansions of $\A_{n,h}$ and $\D_{n,h}$. We assume the model \eqref{eq:linear_model} with $n$ \emph{i.i.d.} samples $\{(\X_i, Y_i)\}$ (the non-i.i.d. case will be discussed in the next subsection). Let $f(\cdot|\X)$ be the conditional density function of the noise $\epsilon$ given $\X$. Further, we define
$
\D=\ep(\X\X'f(0|\X)).
$
{\begin{enumerate}
\item[(C1)]  The conditional density function $f(\cdot|\X)$ is Lipschitz continuous (i.e., $|f(x|\X)-f(y|\X)|\leq L|x-y|$ for any $x,y\in\R$ and some constant $L>0$).  Also, assume that $0<c_{1}\leq \lambda_{\min}(\D)\leq\lambda_{\max}(\D)\leq c_{2}<\infty$ for some constants $c_{1},c_{2}$.

\item[(C2)] Let the smoothing function $H(\cdot)$ satisfy $H(u)=1$ if $u\geq 1$ and $H(u)=0$ if $u\leq -1$. Further, suppose $H(\cdot)$ is twice differentiable and its second derivative $H^{(2)}(\cdot)$ is bounded. Moreover, we assume the bandwidth $h=o(1)$.


\item[(C3)] Assume that $p=o\left(nh/(\log n)\right)$ and $\sup_{\|\thh\|_{2}=1}\ep e^{\eta(\thh'\X)^{2}}<\infty$ for some $\eta>0$.

\item[(C3$^{*}$)]  Assume for some $\kappa>0$, $p=o\left((n^{1-4\kappa}h/\log n)^{1/3}\right)$. Suppose that $\sup_{j}\ep |X_{1,j}|^{a}<\infty$
for some $a\geq 2/\kappa$ and $\sup_{\|\thh\|_2=1}\ep (\thh{'}\X)^{4}<\infty$.

\end{enumerate}
\label{page:condition}

Condition (C1) contains a standard eigenvalue condition related to covariates $\X$ and the smoothness of the conditional density function $f$. 
Condition (C2) is a mild condition on $H$ for smooth approximation.

 Condition (C3) and (C3$^*$) illustrate the relationship between the dimension $p$ and sample size $n$ and the moment condition on covariates $\X$. Either one of them leads to our theoretical results in Propositions \ref{prop1}-\ref{prop2}. As compared to Condition (C3$^*$), Condition (C3) is weaker in terms of the relationship of $p$ and $n$, but requires a stronger moment condition on $\X$.

Under these conditions, we have the following Proposition \ref{prop1} and \ref{prop2} for the asymptotic behavior of $\A_{n,h}$ and $\D_{n,h}$.
\begin{proposition}\label{prop1} Suppose we have an initial estimator $\hat{\be}_{0}$ with $\|\hat{\be}_{0}-\be(\tau)\|_2=O_{\pr}(a_{n})$ with $a_{n}=O(h)$. Assume that (C1), (C2), and (C3) (or (C3$^{*}$)) hold. We have
\begin{eqnarray*}
\Bigl\|\A_{n,h}-\frac{1}{n}\sum_{i=1}^{n}\X_{i}(I\{\epsilon_{i}\geq 0\}+\tau-1)\Bigr\|_2=O_{\pr}\Big{(}\sqrt{\frac{ph\log n}{n}}+a^{2}_{n}+h^{2}\Big{)}.
\end{eqnarray*}
\end{proposition}

\begin{proposition} \label{prop2}  Suppose the conditions in Proposition \ref{prop1} hold.  We have
\begin{eqnarray*}
\|\D_{n,h}-\D\|=O_{\pr}\Big{(}\sqrt{\frac{p\log n}{nh}}+a_{n}+h\Big{)}.
\end{eqnarray*}
\end{proposition}

 Combining Propositions \ref{prop1} and \ref{prop2} with \eqref{eq:key} and with some algebraic manipulations, we have
\begin{eqnarray}\label{eq:rn00}
\hat{\be}-\be(\tau)=\frac{\D^{-1}}{n}\sum_{i=1}^{n}\X_{i}(I\{\epsilon_{i}\geq 0\}+\tau-1)+\r_{n}
\end{eqnarray}
with
\begin{eqnarray}\label{eq:rn0}
\|\r_{n}\|_2=O_{\pr}\Big{(}\sqrt{\frac{p^{2}\log n}{n^{2}h}}+\sqrt{\frac{ph\log n}{n}}+a^{2}_{n}+h^{2}\Big{)}.
\end{eqnarray}

By choosing the bandwidth $h$ shrinking at an appropriate rate, the dominating term of \eqref{eq:rn0} is $a_n^2$, which means that one round of aggregation enables a refinement of the estimator with its bias reducing from $a_n$ to $a_n^2$ (note that $\|\hat{\be}_{0}-\be(\tau)\|_2=O_{\pr}(a_{n})$). Therefore, an iterative refinement of the initial estimator will successively improve the estimation accuracy. The effect of bias reduction is mainly due to the term $\frac{Y_i}{h}H'\Big(\frac{Y_i-\X'_i\widehat{\be}_0}{h}\Big)$ in \eqref{eq:beta_0}, which is induced by the smoothing technique (please see more details in the proof of Proposition \ref{prop1}).

The previous discussions only involve one round of aggregation. } Now we are ready to present the theoretical results for our DC LEQR $\widehat\be^{(q)}$ in Algorithm \ref{algo:dc-leqr} with multiple rounds of aggregations. By a recursive argument based on \eqref{eq:rn00}, we establish the following Bahadur representation.

\begin{theorem}\label{thm:beta} Assume the initial estimator $\hat{\be}_{0}$ in \eqref{ag0}   satisfies $\|\hat{\be}_{0}-\be(\tau)\|_2=O_{\pr}(\sqrt{p/m})$. Let
$h_{g}= \max(\sqrt{p/n},(p/m)^{2^{g-2}})$ for $1\leq g\leq q$. Assuming that (C1), (C2), and (C3) (or (C3$^{*}$)) hold  with $h=h_q$ and $p$ also satisfies $p=O(m/(\log n)^2)$,  then we have
\begin{eqnarray}\label{eq:bah_beta}
\hat{\be}^{(q)}-\be(\tau)=\frac{\D^{-1}}{n}\sum_{i=1}^{n}\X_{i}(I\{\epsilon_{i}\geq 0\}+\tau-1)+\r_{n}
\end{eqnarray}
with
\begin{eqnarray}\label{eq:rn}
\|\r_{n}\|_2=O_{\pr}\Bigl(\sqrt{\frac{ph_{q}\log n}{n}}+\Big{(}\frac{p}{m}\Big{)}^{2^{q-1}}\Bigr).
\end{eqnarray}
\end{theorem}

The classical initial estimator based on QR in \eqref{ag0} will satisfy $\|\widehat{\be}_0-\be(\tau)\|_2=O_{\pr}(\sqrt{p/m})$ under some regularity conditions; see \cite{HE2000120}. According to our choice of bandwidth, we have $h_1=\sqrt{p/m}$ and thus the convergence rate of the initial estimator is $O(h_1)$, which satisfies the condition in Proposition \ref{prop1}.  Furthermore, we note that  any initial estimator $\hat{\be}_{0}$ with $\|\hat{\be}_{0}-\be(\tau)\|_2=O_{\pr}(\sqrt{p/m})$ can be used in the first iteration and the same Bahadur representation in Theorem \ref{thm:beta} holds.

The condition $p=O(m/(\log n)^2)$ in Theorem \ref{thm:beta} ensures that $\sqrt{p/m}=o(1)$, which implies the consistency of the initial estimator (and the $1/(\log n)^2$ factor is used for balancing the terms in  $\|\r_{n}\|_2$ in \eqref{eq:rn0}). 
This condition on  $p$ cannot be implied by either (C3) or (C3$^{*}$); and we choose to present (C3) (or (C3$^{*}$)) at the beginning since they are the minimum requirements to obtain Propositions \ref{prop1}--\ref{prop2}.
On the other hand,   the condition $p=o(nh_q/\log n)$ in (C3) will be satisfied if $p=o(n/(\log n)^2)$ since $h_q \geq \sqrt{p/n}$. Therefore, we can simply impose a stronger condition $p=o(m/(\log n)^{2})$ that unifies the condition $p=O(m/(\log n)^2)$  and that in (C3).

\begin{rem}[Nearly optimal rate of the Bahadur remainder term]\label{rem:optimal}
\emph{From Theorem \ref{thm:beta}, as long as
\begin{equation}\label{eq:q}
q\geq 2 +\log\{\log (\sqrt{p/n})/\log (p/m)\}/\log 2,
\end{equation}
the bandwidth for the $q$-th iteration is $h_q= \sqrt{p/n}$.
Then the first term in the right hand side of  \eqref{eq:rn} is $(p/n)^{3/4} \sqrt{\log n}$ and the second term is bounded by $p/n$, which is dominated by the first term.
Therefore, the Bahadur remainder term $\r_n$ of our method achieves a nearly optimal rate
 \begin{equation}\label{eq:classic_QR}
 \|\r_n\|_2=O_{\pr}\left((p/n)^{3/4}(\log n)^{1/2}\right).
 \end{equation}
In fact, for classical QR estimator $\widehat{\be}_{QR}$ and fixed $p$, 
it is known that the rate $n^{-3/4}$ can not be improved except for a $\log n$ term \citep{koenker2005quantile}. Note that in a common scenario when $n=O(m^A)$ and $p=O(m^{\delta})$ for some constants $A \geq 1$ and $0<\delta<1$, the right hand side of (\ref{eq:q}) is bounded by a constant. Therefore, a constant number of rounds of aggregations is sufficient to obtain a nearly optimal rate in Bahadur representation.}
\end{rem}



Applying the central limit theorem to Theorem \ref{thm:beta}, we obtain the asymptotic distribution of $\widehat{\be}^{(q)}-\be(\tau)$ as follows.

\begin{theorem}\label{thm:clt}
Suppose that all the conditions in Theorem \ref{thm:beta} hold. Further, assume that $n=O(m^{A})$ for some constant $A \geq 1$, $p=o(\min\{n^{1/3}/(\log n)^{2/3},m^{\delta}\})$ for some $0<\delta<1$ and the number of iterations $q$ satisfies \eqref{eq:q}. By choosing the bandwidth sequence $h_{g}= \max(\sqrt{p/n},(p/m)^{2^{g-2}})$ for $1 \leq g \leq q$, for any $\v\in \R^{p}$ with $\v\neq 0$,
\begin{equation}\label{eq:clt}
\frac{n^{1/2}\v'\left(\hat{\be}^{(q)}-\be(\tau)\right)}{\sqrt{\v'\textbf{D}^{-1}\ep [\X\X']\D^{-1}\v}}\Rightarrow N(0,\tau(1-\tau))
\end{equation}
as $m, n\rightarrow \infty$.
\end{theorem}


Note that to establish central limit theorem, we need $p=o(n^{1/3}/(\log  n)^{2/3})$ in Theorem \ref{thm:clt}, which ensures that $\|\r_n\|_2 = o(1/\sqrt{n})$ (see \eqref{eq:rn}). Therefore, as the first term in \eqref{eq:bah_beta} is an average of $n$ \emph{i.i.d.} zero-mean random vectors, the reminder term $\r_n$ is dominated by each coordinate of the first term in \eqref{eq:bah_beta}. For the classical QR estimator $\hat{\be}_{QR}$ in \eqref{eq:QR} (assuming all data is pooled together), the corresponding  condition should be $p=o(n^{1/3}/(\log n)^{2/3})$, see \cite{HE2000120}. This is the same with our condition except for the term $m^{\delta}$ which is required for the consistency of the initial estimator in our method. We also note that the conditions $n=O(m^{A})$ and $p=o(m^{\delta})$ ensure that the number of required iterations $q$ from \eqref{eq:q} is a constant (see Remark \ref{rem:optimal}). Therefore, we only need to perform a constant number of aggregations as $m,n \rightarrow \infty$.



Theorem \ref{thm:clt} shows that $\widehat{\be}^{(q)}$ achieves the same asymptotic efficiency as $\widehat{\be}_{QR}$ in (\ref{eq:QR}) computed directly on all the samples.  When $p$ is fixed, as compared to the na\"ive-DC that also achieves \eqref{eq:clt} but under the condition $n=o(m^2)$, our approach removes the restriction on the relationship of $m$ and $n$ by applying multiple rounds of aggregations. It is also important to  note that the required number of rounds $q$ in \eqref{eq:q} is usually quite small even with a large dimension $p$. 

 Given \eqref{eq:clt}, we only need consistent estimators of $\D$ and $\ep [\X\X']$ to construct confidence interval of $\v'\be(\tau)$ for any given  $\v$. It is  natural to use $\D_{n,h}$ and $\frac{1}{n}\sum_{i=1}^{n}\X_{i}\X'_{i}$ to estimate $\D$ and $\ep [\X\X']$, respectively. These estimators can be easily implemented under memory constraint by averaging the local sample estimators on each batch of data. The proofs of Propositions \ref{prop1}, \ref{prop2} and Theorems \ref{thm:beta}, \ref{thm:clt} are provided in Appendix \ref{sec:proof}.

  \begin{rem}[Data-adaptive choices of bandwidth]\label{rem:adaptive}
\emph{In practice, one could use the bandwidth $h_{g}=c_g\max(\sqrt{p/n},(p/m)^{2^{g-2}})$ with a scaling constant $c_g$ to further improve the empirical performance. An intuitive data-adaptive way of choosing $c_g$ is provided as follows. Given a set of candidate choices for $c_g$ (e.g., $\{c_1,\ldots, c_L\}$), we choose the best $c_g$ by minimizing
     \begin{equation}\label{eq:score}
       S(c):=\Big{\|}\frac{1}{n}\sum_{i=1}^{n}\X_{i}(I\{Y_i-\X_{i}'\hat{\be}^{(g)}_{c}\geq 0\}+\tau-1)\Big{\|}_2,
     \end{equation}
     where $\hat{\be}^{(g)}_{c}$ denotes the estimator in the $g$-th round of aggregation with the constant $c$ in the bandwidth. That is,      $c_g= \argmin_{c\in \{c_1,\ldots, c_L\}} S(c)$. In a distributed setting, the method only requires a small amount of communication. More specifically, given $\hat{\be}^{(g)}_{c}$, each machine $k$ returns $\sum_{i \in \mathcal{H}_k}\X_{i}(I\{Y_i-\X_{i}'\hat{\be}^{(g)}_{c}\geq 0\}+\tau-1)$ (i.e., an $O(p)$ vector) to the center for computing $S(c)$. We also evaluate the performance of our algorithm with the use of data-adaptive bandwidth in Section \ref{sec:adaptive}.}

\emph{It is worthwhile noting that the bandwidth tuning is not a critical issue for our algorithm (in contrast to many other smoothed QR estimators) since our estimator is constructed via multiple rounds of aggregations. Even using an inaccurate constant in bandwidth (as long as the bandwidths for different rounds shrink at the right rate of $(p/m)^{2^{g-2}}$), our method can achieve good performance by simply performing more rounds of aggregations. In Section \ref{sec:adaptive}, we will provide simulation studies to show that our algorithm is insensitive to the scaling constant in the bandwidth.}
\end{rem}

\begin{rem}[Discussions with related literature]\label{rem:comp}
\emph{Note that our estimator $\hat{\be}$ can be written as
{\small
\begin{eqnarray*}
\hat{\be}=\hat{\be}_{0}+\D^{-1}_{n,h} \left[\frac{1}{n}\sum\limits_{i=1}^n \X_i\Big\{H\Big(\dfrac{Y_i-\X'_i\widehat{\be}_0}{h}\Big)+\tau-1+\dfrac{Y_i-\X'_i\hat{\be}_{0}}{h}H'\Big(\dfrac{Y_i-\X'_i\widehat{\be}_0}{h}\Big)\Big\}\right].
\end{eqnarray*}
\normalfont
}
This formula is closely related to the estimator for quantile regression considered in \cite{pang2012variance}, where they introduced the estimator (non-censored version)
\begin{eqnarray*}
\hat{\be}_{PLW}=\hat{\be}_{0}+\A^{-1}_{n} \left[\frac{1}{n}\sum\limits_{i=1}^n \X_i\Big\{H\Big(\dfrac{Y_i-\X'_i\widehat{\be}_0}{\sqrt{\X'_{i}W\X_{i}}}\Big)+\tau-1\Big\}\right],
\end{eqnarray*}
where
$
\A_{n}=\frac{1}{n}\sum\limits_{i=1}^n \frac{\X_i\X_i'}{\sqrt{\X'_{i}W\X_{i}}}H'\Big(\frac{Y_i-\X'_i\widehat{\be}_0}{\sqrt{\X'_{i}W\X_{i}}}\Big)
$
and $W$ is a weight matrix with order  $O(1/n)$.
\cite{pang2012variance} applied the smoothing to the original score function $\sum_{i=1}^{n}\X_{i}(I\{Y_{i}-\X'_{i}\widehat{\be}_{0}\geq 0\}+\tau-1)$, while our estimator comes from the smoothing to the loss function $\rho_{\tau}(\cdot)$ and hence we have an additional term $\frac{Y_i-\X'_i\hat{\be}_{0}}{h}H'\Big(\frac{Y_i-\X'_i\widehat{\be}_0}{h}\Big)$.   Note that this term plays a key role in reducing the bias  induced by the initial estimator $\hat{\be}_{0}$  to $\|\hat{\be}_{0}-\be(\tau)\|^{2}_{2}$. As pointed above, this allows our estimator to have a successive improvement on the estimation accuracy by iteratively updating the initial estimator. 
}

\emph{Moreover, the recent work by  \cite{jordan2016communication} and \cite{wang2016efficient} also proposed iterative approaches in distributed setting for successive refinement of an estimator. However, there are a few key differences between our DC LEQR and their approaches. First, their results
require the loss function to have Lipschitz continuous second-order derivative, which is not satisfied by the original quantile loss function. Even if we replace the indicator function $I\{x\geq 0\}$ in the quantile loss by the smoothed version $H(x/h)$, the second derivative of the loss function will not satisfy their conditions (e.g., Assumption PD in \cite{jordan2016communication} requires that the ``Lipschitz constant'' of the second derivative has a uniform upper bound). Furthermore, our results allow $p\rightarrow \infty$ in the inference problem without $\ell_1$-regularization.}
\end{rem}

\begin{rem}[General heterogenous case]
\emph{We further consider a more general heterogenous case where $\{\X_i,\epsilon_i\}$'s are independent, but not identically distributed from the model \eqref{eq:linear_model}. Due to  space limitations, this heterogenous case is relegated to Appendix \ref{sec:het}.}
\end{rem}

\subsection{Asymptotics for Online LEQR}
The next theorem gives the limiting behavior of the online LEQR in Algorithm \ref{algo:online_LEQR}. 

\begin{theorem}\label{th4.3} Suppose that (C1)-(C3) hold and $p=o(m/(\log m)^{2})$. We have for any $A>0$ and  uniformly in $1\leq j\leq m^{A}$,
\begin{eqnarray}\label{online-LEQR}
\hat{\be}[j]-\be(\tau)&=&\frac{\D^{-1}}{m+j}\sum_{i=-m+1}^{j}\X_{i}\Big{(}I\{\epsilon_{i}\geq 0\}+\tau-1\Big{)}+\r_{m,j},
\end{eqnarray}
where
\begin{eqnarray}\label{rmj}
\|\r_{m,j}\|_2=O_{\pr}\Big(\sqrt{\frac{p}{m+j}}\Big{\{}\Big{(}\frac{p}{m}\Big{)}^{1/4}\sqrt{\log m}+\frac{\sqrt{p}}{m^{1/4}}\Big{\}}\Big).
\end{eqnarray}
Furthermore, when  $p=o(m^{1/4})$, we have for any $\v\in \R^{p}$ with $\v\neq 0$,
\begin{equation}\label{online-LEQR1}
\frac{(m+j)^{1/2}\v'\left(\hat{\be}[j]-\be(\tau)\right)}{\sqrt{\v'\textbf{D}^{-1}\ep [\X\X']\D^{-1}\v}}\Rightarrow N(0,\tau(1-\tau))
\end{equation}
as $m\rightarrow \infty$.
\end{theorem}

We note that $m+j$ is the total number of used samples (including the samples for initialization) up to time $j$.
From (\ref{online-LEQR}),  we have $\|\widehat{\be}[j]-\be(\tau)\|_2=O_{\pr}(\sqrt{p/(m+j)})$ for any $1\leq j\leq m^{A}$  when $\r_{m,j}$ is dominated by the first term. By Theorem 1 in \cite{siegmund1969},
we can further obtain $\|\widehat{\be}[j]-\be(\tau)\|_2=O_{\pr}\Big{(}\sqrt{\frac{p\log\log m}{m+j}}\Big{)}$ uniformly for $1\leq j\leq m^{A}$.  To establish the asymptotic distribution in \eqref{online-LEQR1}, we  need $p=o(m^{1/4})$, which ensures that $\|\r_{m,j}\|_2 = o(1/\sqrt{m+j})$ (see \eqref{rmj}).

\section{Simulations}
\label{sec:simu}
In this section, we provide simulation studies to illustrate the performance of DC LEQR for constructing confidence intervals for  QR  problems. We generate data from a linear regression model
\begin{equation}\label{eq:simu_model}
Y_i=\vec{X}_i'\be+\epsilon_i,\quad i=1,2,\dots,n,
\end{equation}
where  $\X_i=(1, X_{i1}, \ldots, X_{ip})'\in \mathbb{R}^{p+1}$ is a random covariate vector. Here,  $(X_{i1}, \ldots, X_{ip})'$  follows a multivariate uniform distribution $\mathrm{Unif}([0,1]^p)$ with $\text{Corr}(X_{ij},X_{ik})=0.5^{|j-k|}$ for $1\leq j\neq k\leq p$. Please refer to \cite{falk1999simple} for the construction of such a multivariate uniform distribution. The regression coefficient vector $\be=\mathbf{1}_{p+1}$. The errors $\epsilon_i$'s are generated independently from the following three distributions:
\begin{enumerate}[(1)]
\item homoscedastic normal, $\epsilon_i\sim N(0,1)$;
\item heteroscedastic normal, $\epsilon_i\sim N(0,(1+0.3X_{i1})^2)$;
\item exponential, $\epsilon_i\sim\mathrm{Exp}(1)$.
\end{enumerate}
For each quantile level $\tau$, we compute the corresponding true QR coefficient vector $\be(\tau)$ in the QR model \eqref{eq:linear_model} by shifting $\epsilon_i$ such that $\pr(\epsilon_i\leq 0|\X_i)=\tau$:
\begin{enumerate}[(1)]
\item homoscedastic normal, $\be(\tau)=\be+\Phi^{-1}(\tau)\e_1$;
\item heteroscedastic normal, $\be(\tau)=\be+\Phi^{-1}(\tau)\e_1+0.3 \Phi^{-1}(\tau)\e_2$;
\item exponential, $\be(\tau)=\be+F_{\mathrm{exp}}^{-1}(\tau)\e_1$;
\end{enumerate}
Here $\Phi$ and $F_{\mathrm{exp}}$ are the cumulative distribution function of standard normal distribution and exponential distribution with parameter $1$. The vector $\e_i$ (for $i=1,\ldots, p+1$) is the $(p+1)$-dimensional canonical vector with the $i$-th element being one and all the other elements being zero.

We use the integral of a biweight (or quartic) kernel as the smoothing function $H$:
\begin{equation}\label{eq:H}
H(v)=\left\{
\begin{array}{ll}
0&\text{if }v\leq -1,\\
\frac{1}{2}+\frac{15}{16}\left(v-\frac23v^3+\frac15 v^5\right)&\text{if }|v| < 1,\\
1&\text{if }v \geq 1.
\end{array}
\right.
\end{equation}
It is easy to see that it satisfies Condition (C2).

\subsection{Coverage rates}
\label{sec:cover}

\begin{table}[!t]
\centering
    \caption{Coverage rates, bias, and variance of DC LEQR v.s. QR All and na\"ive-DC when $n$ varies from $m^{1.6}$ to $m^{3}$. Noises $\epsilon_i$'s are generated from homoscedastic normal distribution. Dimension $p=15$. Batch size $m=100$. Quantile level $\tau \in \{0.1,0.5,0.9\}$.}
\label{table:bias_1}\label{table:coverage}
\begin{tabular}{ll|crr|crr}
\hline
&$\log_m(n)$&Coverage&Bias ~~~  & Var ~~~~~ &Coverage&Bias~~~  & Var ~~~~~\\
&&Rate&($\times 10^{-2}$) & ($\times 10^{-4}$)&Rate&($\times 10^{-2}$) & ($\times 10^{-4}$)\\ \hline
&          & \multicolumn{3}{c|}{DC LEQR $q=4$} &   \multicolumn{3}{c}{DC LEQR $q=5$} \\
\hline
\multicolumn{2}{l|}{$\tau=0.1$}  &&&&&& \\
           & 1.6 & 0.954  & 0.38  & 20.81 & 0.953    & 0.26  & 21.14 \\
           & 2.0 & 0.956  & 0.13  & 3.04  & 0.953    & 0.09  & 3.05  \\
           & 2.4 & 0.946  & -0.02 & 0.52  & 0.949    & -0.02 & 0.52  \\
           & 3.0 & 0.942  & 0.00  & 0.04  & 0.943    & 0.00  & 0.03  \\
\multicolumn{2}{l|}{$\tau=0.5$}  &&&&&& \\
           & 1.6 & 0.943  & 0.00  & 12.07 & 0.938    & 0.03  & 12.01 \\
           & 2.0 & 0.947  & 0.02  & 1.81  & 0.947    & 0.02  & 1.81  \\
           & 2.4 & 0.944  & 0.00  & 0.29  & 0.945    & 0.00  & 0.29  \\
           & 3.0 & 0.951  & 0.00  & 0.02  & 0.952    & 0.00  & 0.02  \\
\multicolumn{2}{l|}{$\tau=0.9$} &&&&&& \\
           & 1.6 & 0.938  & -0.45 & 21.37 & 0.940    & -0.36 & 21.77 \\
           & 2.0 & 0.942  & -0.05 & 3.65  & 0.932    & -0.02 & 3.65  \\
           & 2.4 & 0.960   & -0.02 & 0.51  & 0.959    & -0.02 & 0.51  \\
           & 3.0 & 0.952  & 0.00  & 0.04  & 0.955    & 0.00  & 0.03  \\
\hline
&&\multicolumn{3}{c|}{QR All} &\multicolumn{3}{c}{Na\"ive-DC}     \\
\hline
\multicolumn{2}{l|}{$\tau=0.1$}  &&&&&& \\
           & 1.6 & 0.948  & 0.15  & 23.04 & 0.638    & 7.86  & 13.82 \\
           & 2.0 & 0.949  & 0.04  & 3.21  & 0.000    & 7.96  & 1.93  \\
           & 2.4 & 0.952  & 0.03  & 0.50  & 0.000    & 7.97  & 0.31  \\
           & 3.0 & 0.953  & 0.01  & 0.03  & 0.000    & 7.95  & 0.02  \\
\multicolumn{2}{l|}{$\tau=0.5$}  &&&&&& \\
           & 1.6 & 0.954  & -0.20 & 11.16 & 0.978    & 0.46  & 8.71  \\
           & 2.0 & 0.951  & -0.01 & 1.68  & 0.968    & 0.40  & 1.32  \\
           & 2.4 & 0.950  & 0.02  & 0.28  & 0.916    & 0.36  & 0.21  \\
           & 3.0 & 0.930  & 0.00  & 0.02  & 0.222    & 0.35  & 0.01  \\
\multicolumn{2}{l|}{$\tau=0.9$}  &&&&&& \\
           & 1.6 & 0.942  & -0.16 & 23.12 & 0.945    & 3.35  & 14.36 \\
           & 2.0 & 0.946  & -0.04 & 3.30  & 0.531    & 3.46  & 2.20  \\
           & 2.4 & 0.947  & 0.02  & 0.52  & 0.000    & 3.45  & 0.35  \\
           & 3.0 & 0.944  & 0.00  & 0.03  & 0.000    & 3.41  & 0.02  \\
\hline
\end{tabular}
\end{table}

\begin{table}[!t]
\centering
    \caption{Coverage rates, bias, and variance of DC LEQR v.s. QR All and na\"ive-DC when $n$ varies from $m^{1.6}$ to $m^{3}$. Noises $\epsilon_i$'s are generated from heteroscedastic normal distribution. Dimension $p=15$. Batch size $m=100$. Quantile level $\tau \in \{0.1,0.5,0.9\}$.}
\label{table:bias_2}\label{table:coverage}
\begin{tabular}{ll|crr|crr}
\hline
&$\log_m(n)$&Coverage&Bias ~~~  & Var ~~~~~ &Coverage&Bias~~~  & Var ~~~~~\\
&&Rate&($\times 10^{-2}$) & ($\times 10^{-4}$)&Rate&($\times 10^{-2}$) & ($\times 10^{-4}$)\\ \hline
&          & \multicolumn{3}{c|}{DC LEQR $q=4$} &   \multicolumn{3}{c}{DC LEQR $q=5$} \\
\hline
\multicolumn{2}{l|}{$\tau=0.1$}  &&&&&& \\
           & 1.6 & 0.941  & 0.62  & 30.11 & 0.940    & 0.56  & 30.24 \\
           & 2.0 & 0.942  & 0.03  & 4.67  & 0.938    & -0.04 & 4.79  \\
           & 2.4 & 0.947  & 0.01  & 0.72  & 0.947    & 0.00  & 0.70  \\
           & 3.0 & 0.952  & 0.00  & 0.06  & 0.950    & 0.00  & 0.05  \\
\multicolumn{2}{l|}{$\tau=0.5$}  &&&&&& \\
           & 1.6 & 0.941  & -0.19 & 16.14 & 0.939    & -0.18 & 16.22 \\
           & 2.0 & 0.946  & 0.02  & 2.48  & 0.944    & 0.02  & 2.48  \\
           & 2.4 & 0.940   & 0.01  & 0.40  & 0.940    & 0.01  & 0.40  \\
           & 3.0 & 0.947  & 0.00  & 0.03  & 0.947    & 0.00  & 0.02  \\
\multicolumn{2}{l|}{$\tau=0.9$} &&&&&& \\
           & 1.6 & 0.925  & -0.61 & 31.55 & 0.926    & -0.48 & 31.46 \\
           & 2.0 & 0.943  & -0.05 & 4.81  & 0.947    & -0.01 & 4.73  \\
           & 2.4 & 0.955  & -0.03 & 0.67  & 0.956    & -0.02 & 0.67  \\
           & 3.0 & 0.953  & -0.01 & 0.09  & 0.957    & 0.00  & 0.04  \\
\hline
&&\multicolumn{3}{c|}{QR All} &\multicolumn{3}{c}{Na\"ive-DC}     \\
\hline
\multicolumn{2}{l|}{$\tau=0.1$}  &&&&&& \\
           & 1.6 & 0.951  & 0.25  & 31.53 & 0.362    & 11.95 & 17.03 \\
           & 2.0 & 0.959  & 0.10  & 4.42  & 0.000    & 11.92 & 2.71  \\
           & 2.4 & 0.937  & -0.01 & 0.79  & 0.000    & 11.88 & 0.44  \\
           & 3.0 & 0.954  & 0.01  & 0.04  & 0.000    & 11.91 & 0.02  \\
\multicolumn{2}{l|}{$\tau=0.5$}  &&&&&& \\
           & 1.6 & 0.950  & -0.09 & 16.70 & 0.976    & 0.16  & 11.99 \\
           & 2.0 & 0.948  & 0.03  & 2.38  & 0.974    & 0.29  & 1.81  \\
           & 2.4 & 0.943  & -0.03 & 0.40  & 0.944    & 0.33  & 0.28  \\
           & 3.0 & 0.960  & 0.00  & 0.02  & 0.320    & 0.36  & 0.02  \\
\multicolumn{2}{l|}{$\tau=0.9$}  &&&&&& \\
           & 1.6 & 0.942  & -0.32 & 31.16 & 0.972    & 2.10  & 19.16 \\
           & 2.0 & 0.946  & -0.13 & 4.68  & 0.882    & 2.00  & 2.94  \\
           & 2.4 & 0.954  & 0.00  & 0.72  & 0.303    & 1.99  & 0.45  \\
           & 3.0 & 0.946  & 0.00  & 0.05  & 0.000    & 2.01  & 0.03  \\
\hline
\end{tabular}
\end{table}

\begin{table}[!t]
\centering
    \caption{Coverage rates, bias, and variance of DC LEQR v.s. QR All and na\"ive-DC when $n$ varies from $m^{1.6}$ to $m^{3}$. Noises $\epsilon_i$'s are generated from exponential distribution. Dimension $p=15$. Batch size $m=100$. Quantile level $\tau \in \{0.1,0.5,0.9\}$.}
\label{table:bias_3}
\begin{tabular}{ll|crr|crr}
\hline
&$\log_m(n)$&Coverage&Bias ~~~  & Var ~~~~~ &Coverage&Bias~~~  & Var ~~~~~\\
&&Rate&($\times 10^{-2}$) & ($\times 10^{-4}$)&Rate&($\times 10^{-2}$) & ($\times 10^{-4}$)\\ \hline
&          & \multicolumn{3}{c|}{DC LEQR $q=4$} &   \multicolumn{3}{c}{DC LEQR $q=5$} \\
\hline
\multicolumn{2}{l|}{$\tau=0.1$}  &&&&&& \\
           & 1.6 & 0.880   & 0.71  & 1.19  & 0.916    & 0.45  & 0.90  \\
           & 2.0 & 0.920   & 0.05  & 0.42  & 0.942    & 0.02  & 0.22  \\
           & 2.4 & 0.915  & 0.00  & 0.13  & 0.931    & 0.02  & 0.09  \\
           & 3.0 & 0.907  & 0.04  & 0.33  & 0.931    & -0.04 & 0.58  \\
\multicolumn{2}{l|}{$\tau=0.5$}  &&&&&& \\
           & 1.6 & 0.933  & -0.12 & 7.76  & 0.931    & -0.09 & 7.75  \\
           & 2.0 & 0.937  & -0.05 & 1.59  & 0.939    & -0.04 & 1.12  \\
           & 2.4 & 0.928  & 0.01  & 0.51  & 0.933    & 0.00  & 0.19  \\
           & 3.0 & 0.934  & 0.00  & 0.39  & 0.941    & -0.01 & 0.04  \\

\multicolumn{2}{l|}{$\tau=0.9$} &&&&&&  \\
           & 1.6 & 0.931  & -0.43 & 61.03 & 0.933    & -0.21 & 58.31 \\
           & 2.0 & 0.908  & 0.00  & 10.75 & 0.916    & 0.12  & 10.37 \\
           & 2.4 & 0.906  & -0.05 & 1.95  & 0.915    & 0.03  & 1.52  \\
           & 3.0 & 0.885  & -0.02 & 0.15  & 0.917    & 0.01  & 0.10  \\
\hline
&&\multicolumn{3}{c|}{QR All} &\multicolumn{3}{c}{Na\"ive-DC}     \\
\hline
\multicolumn{2}{l|}{$\tau=0.1$}  &&&&&& \\
           & 1.6 & 0.945  & 0.14  & 0.87  & 0.422    & 1.98  & 1.11  \\
           & 2.0 & 0.957  & 0.02  & 0.12  & 0.001    & 1.99  & 0.16  \\
           & 2.4 & 0.958  & 0.00  & 0.02  & 0.000    & 1.99  & 0.03  \\
           & 3.0 & 0.944  & 0.00  & 0.00  & 0.000    & 2.00  & 0.00  \\
\multicolumn{2}{l|}{$\tau=0.5$}  &&&&&& \\
           & 1.6 & 0.959  & 0.11  & 7.16  & 0.799    & 3.49  & 5.32  \\
           & 2.0 & 0.944  & 0.02  & 1.15  & 0.070    & 3.46  & 0.93  \\
           & 2.4 & 0.948  & 0.01  & 0.18  & 0.000    & 3.46  & 0.15  \\
           & 3.0 & 0.953  & 0.00  & 0.01  & 0.000    & 3.45  & 0.01  \\
\multicolumn{2}{l|}{$\tau=0.9$}  &&&&&& \\
           & 1.6 & 0.952  & 0.07  & 65.66 & 0.798    & 11.20 & 36.00 \\
           & 2.0 & 0.944  & 0.16  & 10.14 & 0.010    & 11.66 & 5.76  \\
           & 2.4 & 0.953  & 0.00  & 1.59  & 0.000    & 11.68 & 0.96  \\
           & 3.0 & 0.948  & 0.02  & 0.10  & 0.000    & 11.69 & 0.06 \\
\hline
\end{tabular}
\end{table}

We compute the DC LEQR $\widehat{\be}^{(q)}$ in \eqref{ag1} and measure the performance of  $\widehat{\be}^{(q)}$ in terms of the statistical inference. In particular, we report average coverage rates of the confidence interval of $\v_0'\be(\tau)$, where $\v_0=(p+1)^{-1/2}\mathbf{1}_{p+1}$. We set the nominal coverage probability $1- \alpha_0$ to $95\%$.  From Theorem \ref{thm:clt}, an oracle $(1- \alpha_0)$-th confidence interval for $\v_0'\be(\tau)$ is given by
\begin{equation}\label{eq:oracle}
\v_0'\widehat{\be}^{(q)}\pm n^{-1/2}\sqrt{\tau(1-\tau)\v_0'\D^{-1}\ep\left[\vec{X}\vec{X}'\right]\D^{-1}\v_0}z_{\alpha_0/2},
\end{equation}
where $z_{\alpha_0/2}$ is the $(1-\alpha_0/2)$-quantile of the standard normal distribution. To construct the confidence interval, we estimate $\D$ and $\ep\left[\vec{X}\vec{X}'\right]$ by $\D_{n,h}$ and $\frac{1}{n} \sum_{i=1}^n \X_i \X_i'$, respectively. There are two major advantages of this approach. First, since $\D_{n,h}$ has already been obtained in computing DC LEQR, we estimate $\D$ without any extra computation. Second, both $\D_{n,h}$ and $\frac{1}{n} \sum_{i=1}^n \X_i \X_i'$ are in the form of summation over $n$ terms, which can be easily computed in a distributed setting with little communication cost. As we will show in Table \ref{table:sesd_adaptive}, the proposed estimator is very close to the truth. The scaling constant in bandwidth is simply set to one (as in our theorems) and more detailed experiments on the sensitivity analysis of the scaling constant is provided in Section \ref{sec:adaptive}. We report  empirical coverage rates as an average of 1000 independent runs of the simulations.

In Tables \ref{table:bias_1}--\ref{table:bias_3}, we present the empirical coverage rates of our DC LEQR estimator, the na\"ive-DC estimator, and the oracle QR estimator in \eqref{eq:QR} computed on all data points (denoted as QR All) for three different noise models.  More precisely, we generate the error from one of three distributions (i.e., homoscedastic normal for Table \ref{table:bias_1}, heteroscedastic normal for Table \ref{table:bias_2}, and exponential for Table \ref{table:bias_3}) and consider three different quantile levels $\tau=0.1,0.5,0.9$.
In our experiment, we set $m=100$,  $p=15$ and vary $n$ form $m^{1.6}$ to $m^3$ (i.e., $\log_m(n)$ from 1.6 to 3).  From \eqref{eq:q}, it is easy to see that we need number of aggregations $q \geq 4$. Thus, we report the performance of DC LEQR $\widehat{\be}^{(q)}$ for $q=4$ and $q=5$.   We also report the case of $p=3$ in Appendix \ref{sec:p_15}. 

As one can see from Tables \ref{table:bias_1}--\ref{table:bias_3}, for most of the settings, the coverage rates of our DC LEQR are  close to the nominal level of $95\%$ after $4$ rounds of aggregations ($q=4$). The coverage performance becomes quite stable for $q=5$ iterations. On the other hand, for the na\"ive-DC estimator, the coverage rates are quite low in most settings, especially when $n$ is  larger than $m^2$.

Note that for na\"ive-DC and QR All, we use the same estimator of the limiting variance in \eqref{eq:oracle} as in our DC LEQR when $q=4$. More precisely, we use $\D_{n,h}$ computed in the $4$-th iteration to estimate $\D$ in \eqref{eq:oracle} when constructing the confidence intervals of na\"ive-DC and QR All estimators. We will show in Table \ref{table:sesd_adaptive} below that the proposed estimator of the limiting variance performs well. In fact, we also use the true limiting variance to construct the confidence interval and the coverage rates for all the methods are almost the same.    

In  addition,  we  also  report  the  simulation study  with  large  dimension $(p=  1000)$. The  results  and  analysis  are relegated  to Appendix \ref{sec:large_p}. From the results, we can infer that the  coverage rates get better as the iterative refinement proceeds. In particular, the coverage rates are close to the nominal level 95\% after $4$ iterations when the dimension $p=1000$. In summary, when the dimension $p$ is large, the proposed DC LEQR algorithm still achieves desirable performance with a small number of iterations.



\subsection{Bias and variance analysis}\label{sec:bv}
To see the improvement of DC LEQR over na\"ive-DC when $n$ is excessively larger than the subset size $m$, we also report the mean bias and variance of our proposed DC LEQR, na\"ive-DC, and QR All in Tables \ref{table:bias_1}--\ref{table:bias_3}. The mean bias and variance of $\v_0'\widehat{\be}$ are based on 1000 independent runs of simulations.


From Tables \ref{table:bias_1}--\ref{table:bias_3}, the bias of our method is quite small while the na\"ive-DC approach has a much larger bias regardless of the sample size $n$. For the variance, it decays with  the rate $1/n$ as $n$ goes large for all methods.
For most cases of using na\"ive-DC, as $\log_m(n)$ exceeding $2$, the squared bias becomes comparable or larger than the variance, which explains the reason of the failure of na\"ive-DC when $n$ is large as compared to $m$. On the other hand,  the bias of our proposed DC LEQR is similar to that of QR All and much smaller than that of na\"ive-DC.

\subsection{Sensitivity analysis and data-adaptive choice of the bandwidth}\label{sec:adaptive}
In this section, we show the empirical performance of the data-adaptive choice of bandwidth in Remark \ref{rem:adaptive} and the sensitivity of  the scaling constant in bandwidth. Due to space limitations, we report $\tau=0.1$, homoscedastic normal noise case as an example.  More noise cases (e.g., heteroscedastic normal and exponential cases) are relegated to Appendix \ref{sec:supp_sens}, and observations are similar to the homoscedastic normal case.
\begin{table}[!t]
\centering
\caption{Coverage rates for DC LEQR with iterations $q=1,2,3,4,5$ for different choices of the scaling constant $c$. Noises $\epsilon_i$'s are generated from homoscedastic normal distribution. Dimension $p=15$. Batch size $m=100$.  Sample size $n$ varies from $m^{1.6}$ to $m^{3}$.  Quantile level $\tau=0.1$.}
\label{table:adaptive}
\begin{tabular}{lc|ccccc}
\hline
        & $\log_m(n)$ & $q=1$ & $q=2$ & $q=3$ & $q=4$ & $q=5$\\
\hline
$c=1$    &     &          &          &          &      &     \\
 & 1.6 & 0.642 & 0.909 & 0.949 & 0.954 & 0.953\\
 & 2.0 & 0.311 & 0.911 & 0.956 & 0.956  &0.953\\
 & 2.4 & 0.077 & 0.427 & 0.877 & 0.946 &0.949\\
 & 3.0 & 0.000 & 0.231 & 0.611 & 0.942 &0.943\\
\hline
$c=3$    &     &          &          &  &        &          \\
 & 1.6 & 0.711 & 0.907 & 0.951 & 0.949 &0.944\\
 & 2.0 & 0.303 & 0.891 & 0.947 & 0.946 &0.946\\
 & 2.4 & 0.111 & 0.521 & 0.889 & 0.952 &0.949\\
 & 3.0 & 0.000 & 0.306 & 0.687 & 0.953 &0.950\\
\hline
$c=5$    &     &          &          &          &          \\
 & 1.6 & 0.644 & 0.900 & 0.959 & 0.950 &0.950\\
 & 2.0 & 0.196 & 0.849 & 0.942 & 0.952 &0.951\\
 & 2.4 & 0.079 & 0.469 & 0.815 & 0.944 &0.944\\
 & 3.0 & 0.000 & 0.120 & 0.588 & 0.947 &0.949\\
\hline
$c=10$   &     &          &          &          &          \\
 & 1.6 & 0.597 & 0.814 & 0.955 & 0.949 &0.947\\
 & 2.0 & 0.129 & 0.729 & 0.944 & 0.951 &0.951\\
 & 2.4 & 0.000 & 0.402 & 0.777 & 0.952 &0.949\\
 & 3.0 & 0.000 & 0.011 & 0.513 & 0.950 &0.946\\
\hline
data-adaptive &     &          &          &          &          \\
 & 1.6 & 0.724 & 0.929 & 0.946 & 0.949 &0.948\\
 & 2.0 & 0.336 & 0.914 & 0.947 & 0.954 &0.949\\
 & 2.4 & 0.187 & 0.564 & 0.912 & 0.941 &0.944\\
 & 3.0 & 0.000 & 0.385 & 0.710 & 0.954 &0.951\\
\hline
\end{tabular}
\end{table}

Table \ref{table:adaptive} shows  coverage rates of the DC LEQR with $q=1,2,3,4,5$ iterations. Similar to the setting in Tables \ref{table:bias_1}--\ref{table:bias_3}, we choose $m=100$,  $p=15$, and $n$ varies from $m^{1.6}$ to $m^3$. We report the performance of DC LEQR using different fixed constants $c=1,3,5,10$ in bandwidth $h_g$ for $1\leq g \leq q$ (using the same constant in all iterations) as well as our data-adaptive choice of bandwidth. For the data-adaptive bandwidth, we choose the best scaling constant from a list of 1000 equally spaced constants from a very small number (0.1) to a large one (100) according to Remark \ref{rem:adaptive}.   Note that 
different scaling constants will be chosen for different iterations. As one can see, for $q=1$ and $q=2$, the adaptive method indeed achieves better coverage than other scaling constants. On the other hand, when $q \geq 3$, all different choices of scaling constants lead to coverage rates close to the nominal level of 95\%. This experiment suggests that for our proposed iterative aggregation approach, even when a sub-optimal scaling constant is used, one can still achieve good performance by performing more iterations.
\begin{table}[!t]
\centering
\caption{Square root of the ratio of the estimated variance and the true limiting variance  using  different choices of the scaling constant $c$ in bandwidths. Noises $\epsilon_i$'s are generated from homoscedastic normal distribution. Dimension $p=15$. Batch size $m=100$.  Sample size $n$ varies from $m^{1.6}$ to $m^{3}$.  Quantile level $\tau=0.1$. }
\label{table:sesd_adaptive}
\begin{tabular}{lc|ccccc}
\hline
        & $\log_m(n)$ & $q=1$ & $q=2$ & $q=3$ & $q=4$ &$q=5$\\
\hline
$c=1$    &     &          &          &          &     &     \\
 & 1.6 & 1.05 & 1.10 & 1.11 & 1.07 &1.05\\
 & 2.0 & 1.14 & 1.09 & 0.99 & 1.04 &1.03\\
 & 2.4 & 1.27 & 1.24 & 1.12 & 0.99 &1.00\\
 & 3.0 & 1.12 & 1.07 & 1.03 & 1.00 &1.01\\
\hline
$c=3$    &     &          &          &  &        &          \\
 & 1.6 & 1.14 & 1.01 & 0.99 & 0.99 &1.00\\
 & 2.0 & 1.06 & 1.04 & 1.02 & 1.00 &1.01\\
 & 2.4 & 1.08 & 1.03 & 1.03 & 1.01 &1.01\\
 & 3.0 & 1.04 & 1.00 & 0.99 & 0.99 &0.99\\
\hline
$c=5$    &     &          &          &          &          \\
 & 1.6 & 1.09 & 1.04 & 0.99 & 0.98 &0.99\\
 & 2.0 & 1.11 & 1.12 & 1.07 & 1.03 &1.04\\
 & 2.4 & 1.04 & 1.07 & 1.00 & 1.00&1.01 \\
 & 3.0 & 1.07 & 1.02 & 1.01 & 1.01 &1.01\\
\hline
$c=10$   &     &          &          &          &          \\
 & 1.6 & 0.74 & 0.84 & 0.89 & 0.94 &0.96\\
 & 2.0 & 0.86 & 0.86 & 0.92 & 0.96 &0.99\\
 & 2.4 & 0.90 & 0.91 & 0.96 & 0.99 &1.01\\
 & 3.0 & 0.88 & 0.90 & 0.98 & 1.00 &1.00\\
\hline
data-adaptive &     &          &          &          &          \\
 & 1.6 & 1.01 & 1.06 & 1.03 & 1.01 &1.00\\
 & 2.0 & 1.03 & 1.03 & 1.01 & 0.98 &0.97\\
 & 2.4 & 1.01 & 1.06 & 1.01 & 1.00 &1.01\\
 & 3.0 & 1.01 & 1.04 & 1.02 & 1.01 &1.01\\
\hline
\end{tabular}
\end{table}

 Moreover, to investigate the sensitivity of the scaling constant in terms of  variance estimation, we also present the square root of the ratio of the estimated variance of our approach versus the true limiting variance, i.e., 
\begin{equation}\label{eq:ratio}
  \frac{\sqrt{\v_0'\D_{n,h}^{-1}\frac{1}{n}\sum_{i=1}^n\left[\X_i\X'_i\right]\D_{n,h}^{-1}\v_0}}{\sqrt{\v_0'\D^{-1}\ep\left[\vec{X}\vec{X}'\right]\D^{-1}\v_0}},
\end{equation}
where $\D_{n,h}$ is computed for iterations $q=1,2,3,4,5$ and for each fixed $q$, the bandwidth $h_g=c \max(\sqrt{p/n},(p/m)^{2^{g-2}})$ for $1 \leq g \leq q$. In Table \ref{table:sesd_adaptive}, we report the performance of the variance estimation with different choices of the scaling constant $c$ of the bandwidth $h$ in constructing $\D_{n,h}$ .

From Table \ref{table:sesd_adaptive}, the ratio is very close to $1$ when $q=4$ or $5$. Therefore, when $q$ is large, the proposed variance estimator is a reliable one in the distributed setting. Moreover, we notice  that the ratio is very stable for different choices of the scaling constant $c$, which illustrates the robustness of the estimator.


\begin{table}[!t]
\caption{Bias ($\times 10^{-2})$, variance ($\times 10^{-4})$, coverage rates (nominal level $95\%$) and computation time ($\times 100$ seconds) of DC LEQR for different $q$ versus the standard QR estimator on the entire data (QR All), and na\"ive-DC. Noises $\epsilon_i$'s are generated from homoscedastic normal distribution. Dimension $p=15$. Quantile level $\tau=0.1$. }
\label{table:time}
\centering
\begin{tabular}{llrrrrrr}
\hline
                               &          & \multicolumn{4}{c}{DC LEQR}   & QR~~    & Na\"ive~  \\
                               &          & $q=1$   & $q=2$   & $q=3$   & $q=4$   &  All~~   & DC~~  \\
\hline
\multicolumn{2}{l}{$m=100$, $n=10^6$}         &          &          &          &          &          \\
 & Bias & 6.112 & 0.865 & 0.038 & -0.020 & -0.020 & 7.947\\
 & Variance & 35.464 & 1.637 & 0.037 & 0.035 & 0.031 & 0.24\\
 & Coverage & 0.001 & 0.314 & 0.940 & 0.942 & 0.942 & 0.000\\
 & Time & 0.409 & 0.821 & 1.233 & 1.643 & 8.015 & 2.421\\
\hline
\multicolumn{2}{l}{$m=500$, $n=10^6$}         &          &          &          &          &          \\
 & Bias & 1.334 & 0.029 & -0.008 & -0.010 & -0.009 & 0.214\\
 & Variance & 0.642 & 0.042 & 0.029 & 0.029 & 0.029 & 0.029\\
 & Coverage & 0.087 & 0.914 & 0.947 & 0.951 & 0.951 & 0.132\\
 & Time & 0.499 & 0.993 & 1.488 & 1.982 & 7.909 & 2.549\\
\hline
\multicolumn{2}{l}{$m=500$, $n=10^7$}         &          &          &          &          &          \\
 & Bias & 1.171 & 0.046 & -0.005 & -0.005 & -0.005 & 0.237\\
 & Variance & 0.885 & 0.016 & 0.002 & 0.002 & 0.002 & 0.002\\
 & Coverage & 0.024 & 0.642 & 0.943 & 0.948 & 0.947 & 0.000\\
 & Time & 4.555 & 9.106 & 13.648 & 18.188 & 143.521 & 25.289\\	
\hline
\multicolumn{2}{l}{$m=1000$, $n=10^7$}          &          &          &          &          &          \\
 & Bias & 0.786 & 0.016 & -0.007 & -0.007 & -0.007 & 0.140\\
 & Variance & 0.167 & 0.002 & 0.003 & 0.003 & 0.003 & 0.003\\
 & Coverage & 0.014 & 0.897 & 0.952 & 0.947 & 0.947 & 0.013\\
 & Time & 4.547 & 9.087 & 13.625 & 18.164 & 137.425 & 25.353\\
\hline
\end{tabular}
\end{table}

\subsection{Computation efficiency}
\label{sec:comput}
We further conduct experiments to illustrate the computation efficiency of our algorithm for different $m$ and $n$ with $\tau=0.1$, $p=15$ and $\epsilon_i \sim N(0,1)$. We compare the computation time of DC LEQR versus that of na\"ive-DC as well as QR All in Table \ref{table:time}. We  report the bias and variance and the coverage rates for reference of the performance of the estimators.

First of all, the computation time of DC  LEQR is about as twice faster than that of the QR All, especially when $n$ is large. It is also faster than the na\"ive-DC and with a much better coverage when $n$ is much larger than $m$. Moreover, the time of our algorithm grows almost linearly in both $n$ and $q$, which is consistent with the computation time analysis in Section \ref{sec:est}. In contrast, we observe that the time of QR grows faster than a linear function in the sample size $n$. We also observe that, for each fixed $n$, the value $m$ has little effect on the computation time of DC LEQR. For na\"ive-DC, the squared bias in these cases dominate the variance, so the coverage rates are far below the nominal level. In the meantime, DC LEQR has around $95\%$ coverage in all four cases after $2$ iterations, and shows a similar behavior of bias and variance as  in Table \ref{table:bias_1}.

Recently, researchers have developed new optimization techniques based on  alternating direction method of multiplier (ADMM)  for solving QR problems (see, e.g., \cite{yu2017parallel,gu2017admm}). We further conduct comparisons to the ADMM approach and the details are provided in Appendix \ref{sec:admm}.

Finally, we also conduct simulation studies for online LEQR and the results are presented in Appendix \ref{sec:online_exp}.

\section{Conclusions and future works}
\label{sec:con}

In this paper, we propose a novel inference approach for quantile regression under the memory constraint. The proposed method achieves the same asymptotic efficiency as the quantile regression estimator using all the data. Furthermore, it allows a weak condition on the sample size $n$ as a function of memory size $m$ and is computationally attractive. One  key insight from this work is that na\"ively splitting data and averaging local estimators could be sub-optimal. Instead, the iterative refinement idea can lead to much improved performance for some inference problems in distributed environments.


In some applications, one would expect a weaker assumption on the distribution of data where the data could be correlated.  It would be an interesting future direction to study the problem of inference for correlated data in distributed settings.  Moreover, for the online problem, it is also interesting to consider the case where the model (e.g., $\be(\tau)$) is evolving over time. In this case, some exponential decaying techniques to down weight historical data might be useful.


In the future, we would also like to further explore this idea to other QR problems under memory constraints or in a distributed setup, e.g., $\ell_1$-penalized high-dimensional quantile regression (see, e.g., \cite{belloni2011l1,wang2012quantile,fan2016multitask}) and censored quantile regression (see, e.g., \cite{wang2009locally,kong2013global,volgushev2014censored,Leng:14,He:17:censored}).

\section*{Acknowledgements}

The authors are very grateful to three anonymous referees and the associate editor for their detailed and constructive comments that considerably improved the quality of this paper.

\bibliographystyle{imsart-nameyear}
\bibliography{ref}
\newpage
\appendix

\section{Proof of technical results}
\label{sec:proof}
\subsection{Proof of Proposition \ref{prop1}}
Let
\begin{eqnarray*}
	\C_{n,h}&=&\A_{n,h}-\frac{1}{n}\sum_{i=1}^{n}\X_{i}(I\{\epsilon_{i}\geq 0\}+\tau-1)\cr
	&=&\dfrac{1}{n}\sum\limits_{i=1}^n \X_i\left\{H\left(\dfrac{Y_i-\X_i'\widehat{\be}_0}{h}\right)-I\{\epsilon_{i}\geq 0\}+\dfrac{\epsilon_i}{h}H'\left(\dfrac{Y_i-\X_i'\widehat{\be}_0}{h}\right)\right\}.
\end{eqnarray*}
Note that $\|\C_{n,h}\|_2=\sup_{\v\in \mathbb{R}^{p}, \|\v\|_2=1}|\v'\C_{n,h}|$.

Let $S^{p-1}_{1/2}$ be a $1/2$ net of the unit sphere $S^{p-1}$ in the Euclidean distance in $\mathbb{R}^{p}$. By the proof of Lemma 3 in \cite{cai2010optimal}, we have $d_{p}:=$Card$(S^{p-1}_{1/2})\leq 5^{p}$. Let $\v_{1},...,\v_{d_{p}}$ be the centers of the $d_{p}$ elements in the net. Therefore for any $\v$ in $S^{p-1}$, we have $\|\v-\v_{j}\|_2\leq 1/2$ for some $j$. Therefore, $\|\C_{n,h}\|_2\leq \sup_{j\leq d_{p}}|\v_{j}'\C_{n,h}|+\|\C_{n,h}\|_2/2$. That is,  $\|\C_{n,h}\|_2\leq 2\sup_{j\leq d_{p}}|\v_{j}'\C_{n,h}|$.

For $\al\in\R^p$, denote
\begin{eqnarray*}
	C_{n,h,j}(\al)&=&\dfrac{1}{n}\sum\limits_{i=1}^n \v_{j}'\X_i\left\{H\left(\dfrac{Y_i-\X_i'\al}{h}\right)-I\{\epsilon_i\geq 0\}+\dfrac{\epsilon_i}{h}H'\left(\dfrac{Y_i-\X_i'\al}{h}\right)\right\}\cr
	&=:&\dfrac{1}{n}\sum\limits_{i=1}^n \v_{j}'\X_i H_{h}(\al).
\end{eqnarray*}
When $\|\hat{\be}_{0}-\be(\tau)\|_2\leq a_{n}$, we have $\|\C_{n,h}\|_2\leq 2\sup_{j\leq d_{p}}\sup_{\|\al-\be(\tau)\|_2\leq a_{n}}|C_{n,h,j}(\al)|.$
Denote $\be(\tau)=(\beta_{1},...,\beta_{p})'$. For every $i$, we divide the interval $[\beta_{i}-a_{n},\beta_{i}+a_{n}]$ into $n^{M}$ small subintervals and each has length $2a_{n}/n^{M}$, where $M$ is a large positive number. Therefore, there exists a set of points in $\R^{p}$, $\{\al_{k}, 1\leq k\leq n^{Mp}\}$, such that for any $\al$ in the ball $\|\al-\be(\tau)\|_2\leq a_{n}$, we have $\|\al-\al_{k}\|_2\leq 2\sqrt{p}a_{n}/n^{M}$ for some $1\leq k\leq n^{Mp}$.
Let $\Delta (\al)=\al-\be(\tau)$. We have
\begin{eqnarray}\label{pr2}
	&&\dfrac{\epsilon_i}{h}H'\left(\dfrac{Y_i-\X_i'\al}{h}\right) \cr
	&&\quad=\dfrac{\epsilon_i-\X'_{i}\Delta(\al)}{h}H'\Big(\dfrac{\epsilon_i-\X_i'\Delta(\al)}{h}\Big)+\dfrac{\X'_{i}\Delta(\al)}{h}H'\Big(\dfrac{\epsilon_i-\X_i'\Delta(\al)}{h}\Big).
\end{eqnarray}
Note that $xH'(x)$ and $H'(x)$ are Lipschitz continuous and $H'(x)$ is bounded. We have for $\|\al-\be(\tau)\|_2\leq a_{n}$,
\begin{eqnarray*}
	&&\Big{|}\dfrac{\epsilon_i}{h}H'\left(\dfrac{Y_i-\X_i'\al}{h}\right)-\dfrac{\epsilon_i}{h}H'\left(\dfrac{Y_i-\X_i'\al_{k}}{h}\right)\Big{|}\cr
	&&\quad\leq Ch^{-1}\|\X_{i}\|_2\|\al-\al_{k}\|_2+Ch^{-2}\|\X_{i}\|_2^{2}\|\Delta(\al)\|_2\|\al-\al_{k}\|_2\cr
	&&\quad\leq Ch^{-1}\sqrt{p}a_{n}n^{-M}\|\X_{i}\|_2+Ch^{-2}\sqrt{p}a^{2}_{n}n^{-M}\|\X_{i}\|_2^{2}.
\end{eqnarray*}
Similarly,
\begin{eqnarray*}
	\Big{|}H\left(\dfrac{Y_i-\X_i'\al}{h}\right)-H\left(\dfrac{Y_i-\X_i'\al_{k}}{h}\right)\Big{|}\leq Ch^{-1}\sqrt{p}a_{n}n^{-M}\|\X_{i}\|_2.
\end{eqnarray*}
Therefore, we have
\begin{eqnarray*}
	&&\sup_{j}\sup_{\|\al-\be(\tau)\|_2\leq a_{n}}|C_{n,h,j}(\al)|-\sup_{j}\sup_{k}|C_{n,h,j}(\al_{k})|\cr
	&&\quad\leq C\frac{\sqrt{p}a_{n}}{n^{M+1}h}\sum_{i=1}^{n}\|\X_{i}\|_2^{2}+
	C\frac{\sqrt{p}a^{2}_{n}}{n^{M+1}h^{2}}\sum_{i=1}^{n}\|\X_{i}\|_2^{3}.
\end{eqnarray*}
Since $\max_{i,j}\ep|X_{i,j}|^{3}<\infty$, by letting $M$ large enough, we have
\begin{eqnarray}\label{pr1}
	\sup_{j}\sup_{\|\al-\be\|_2\leq a_{n}}|C_{n,h,j}(\al)|-\sup_{j}\sup_{k}|C_{n,h,j}(\al_{k})|=O_{\pr}(n^{-\gamma})
\end{eqnarray}
for any $\gamma>0$.
It is enough to show that $\sup_{j}\sup_{k}|C_{n,h,j}(\al_{k})|$ satisfies the bound in the proposition.

We first prove the proposition under (C3$^{*}$).
Let $\X_{i}=(X_{i,1},..,X_{i,p})'$, $\hat{X}_{i,j}=X_{i,j}I\{|X_{i,j}|\leq n^{\kappa}\}$ and $\hat{\X}_{i}=(\hat{X}_{i,1},...,\hat{X}_{i,p})'$.
Then
\begin{eqnarray*}
	\sup_{\|\thh\|_2=1}\ep (\thh'\hat{\X}_{i})^{4}&\leq&8\sup_{\|\thh\|_2=1}\ep (\thh'\X_{i})^{4}+8\sup_{\|\thh\|_2=1}\Big(\ep\sum\limits_{j=1}^p |\theta_jX_{i,j}|I\{|X_{i,j}|\geq n^{\kappa}\}\Big)^4\cr
	&\leq& 8\sup_{\|\thh\|_2=1}\ep (\thh'\X_{i})^{4}+8p\sum_{j=1}^{p}\ep (X_{i,j})^{4}I\{|X_{i,j}|\geq n^{\kappa}\}\leq C.
\end{eqnarray*}
Define
{\small\begin{eqnarray}\label{eq:hha}
		\hat{C}_{n,h,j}(\al)&=&\dfrac{1}{n}\sum\limits_{i=1}^n \v_{j}'\hat{\X}_i\left\{H\Big(\dfrac{\epsilon_i-\hat{\X}_i'\Delta(\al)}{h}\Big)-I\{\epsilon_i\geq 0\}+\dfrac{\epsilon_i}{h}H'\Big(\dfrac{\epsilon_i-\hat{\X}_i'\Delta(\al)}{h}\Big)\right\}\cr
		&=:&\dfrac{1}{n}\sum\limits_{i=1}^n \v_{j}'\hat{\X}_i\hat{H}_{h}(\al).
	\end{eqnarray}
}%
We have
\begin{eqnarray}
	\label{eq:p_n_ratio}
	\pr\Big{(}\sup_{j}\sup_{k}|C_{n,h,j}(\al_{k})|\neq \sup_{j}\sup_{k}|\hat{C}_{n,h,j}(\al_{k})|\Big{)}&\leq& \sum_{i=1}^{n}\sum_{j=1}^{p}\pr(|X_{i,j}|\geq n^{\kappa}) \cr
	\leq \sum_{i=1}^{n}\sum_{j=1}^{p}\dfrac{\ep(|X_{i,j}|^{2/\kappa})}{n^2}&=&O(p/n)=o(1).
\end{eqnarray}
Let $\ep_{*}(\cdot)$ denote the conditional expectation given $\{\X_{i},1\leq i\leq n\}$. By Condition (C1), since the conditional density function $f(\cdot|\X_{i})$ is Lipschitz continuous, it is bounded. Also by Condition (C2) that $H(x)$ is bounded with support on $[-1,1]$, we have
\begin{eqnarray}
	&&\ep_{*}\Big{[}H\Big(\dfrac{\epsilon_i-\widehat{\X}_i'\Delta(\al)}{h}\Big)-I\{\epsilon_i\geq 0\}\Big{]}^{2}\cr
	&&=h\int_{-\infty}^{\infty}\Big{[}H(x)-I\{x\geq -\hat{\X}'_{i}\Delta(\al)/h\}\Big{]}^{2}f(hx+\hat{\X}'_{i}\Delta(\al)|\X_{i})dx\cr
	&&\leq Ch(|\hat{\X}'_{i}\Delta(\al)/h|+1)\label{eq:e_hh_sq}
\end{eqnarray}
and
\begin{eqnarray}
	&&\ep_{*}\Big{[}\dfrac{\epsilon_i}{h}H'\Big(\dfrac{\epsilon_i-\hat{\X}_i'\Delta(\al)}{h}\Big)\Big{]}^{2}\cr
	&&=h\int_{-\infty}^{\infty}\{(x+\hat{\X}'_{i}\Delta(\al)/h)H'(x)\}^{2}f(hx+\hat{\X}'_{i}\Delta(\al)|\X_{i})dx\cr
	&&\leq Ch(|\hat{\X}'_{i}\Delta(\al)/h|^{2}+1).\label{eq:e_h'_sq}
\end{eqnarray}
By \eqref{eq:e_hh_sq}, \eqref{eq:e_h'_sq} and the definition of $\widehat{H}_h(\al)$ in \eqref{eq:hha}, we have\label{page:ineq}
\begin{eqnarray*} \ep(\v_{j}'\hat{\X}_i\hat{H}_{h}(\al))^{2}&\leq&2\ep(\v_{j}'\hat{\X}_i)^2\bigg[H\Big(\dfrac{\epsilon_i-\hat{\X}_i'\Delta(\al)}{h}\Big)-I\{\epsilon_i\geq 0\}\bigg]^2\cr
	&&+2\ep(\v_{j}'\hat{\X}_i)^2\bigg[\dfrac{\epsilon_i}{h}H'\Big(\dfrac{\epsilon_i-\hat{\X}_i'\Delta(\al)}{h}\Big)\bigg]^2\cr
	&\leq& Ch\Big{(}\sup_{\|\thh\|_2=1}\ep(\thh'\hat{\X}_{i})^{2}+\sup_{\|\thh\|_2=1}\ep|\thh'\hat{\X}_{i}|^{3}\|\al-\be(\tau)\|_2/h\cr
	&&\qquad+\sup_{\|\thh\|_2=1}\ep(\thh'\hat{\X}_{i})^{4}\|\al-\be(\tau)\|_2^{2}/h^{2}\Big{)}\cr
	&\leq&Ch\Big{(}1+\|\al-\be(\tau)\|_2/h+\|\al-\be(\tau)\|_2^{2}/h^{2}\Big{)},
\end{eqnarray*}
where we used the inequalities that
\begin{eqnarray*}
	\sup_{\|\v\|=1,\|\u\|=1}\ep\{(\v^{'}\hat{\X}_{i})^{2}|\u^{'}\hat{\X}_{i}|\}\leq \sup_{\|\thh\|=1}\ep|\thh^{'}\hat{\X}_{i}|^{3}\leq C,\cr
	\sup_{\|\v\|=1,\|\u\|=1}\ep\{(\v^{'}\hat{\X}_{i})^{2}(\u^{'}\hat{\X}_{i})^{2}\}\leq \sup_{\|\thh\|=1}\ep|\thh^{'}\hat{\X}_{i}|^{4}\leq C.
\end{eqnarray*}
By (\ref{pr2}), noting that $H$ and $xH'(x)$ is bounded, we can obtain that
\begin{eqnarray*}
	|\v_{j}'\hat{\X}_i\hat{H}_{h}(\al)|&\leq& C\|\hat{\X}_{i}\|_2(1+\|\hat{\X}_{i}\|_2\|\al-\be(\tau)\|_2/h)\cr
	&\leq& C\sqrt{p}n^{\kappa}+Cpn^{2\kappa}\|\al-\be(\tau)\|_2/h.
\end{eqnarray*}
By Condition (C3$^*$),
\begin{eqnarray*}
	p\log n=o\Big{(}\frac{\sqrt{nph\log n}}{pn^{2\kappa}}\Big{)}.
\end{eqnarray*}
By Bernstein's inequality, we can get for any $\gamma>0$, there exists a constant $C$ such that
\begin{eqnarray*}
	\sup_{j}\sup_{k}\pr\Big{(}|\hat{C}_{n,h,j}(\al_{k})-\ep \hat{C}_{n,h,j}(\al_{k})|\geq C\sqrt{\frac{p h\log n}{n}}\Big{)}=O(n^{-\gamma p}).
\end{eqnarray*}
This yields that
\begin{eqnarray}\label{pr3}
	\sup_{j}\sup_{k}|C_{n,h,j}(\al_{k})-\ep C_{n,h,j}(\al_{k})|=O_{\pr}\Big{(}\sqrt{\frac{p h\log n}{n}}\Big{)}.
\end{eqnarray}
It remains to give a bound for $\ep C_{n,h,j}(\al)$. Let $F(x|\X_{i})$ be the conditional distribution of $\epsilon_i$ given $\X_{i}$. We have
\begin{eqnarray*}
	\ep_{*}H\Big(\dfrac{\epsilon_i-\hat{\X}_i'\Delta(\al)}{h}\Big)&=&h\int_{-\infty}^{\infty}H(x)f(hx+\hat{\X}_i'\Delta(\al)|\X_{i})dx\cr
	&=&-\int_{-\infty}^{\infty}H(x)d(1-F(hx+\hat{\X}_i'\Delta(\al)|\X_{i}))\cr
	&=&\int_{-1}^{1}(1-F(hx+\hat{\X}_i'\Delta(\al)|\X_{i})) H'(x)dx\cr
	&=&1-\tau-f(0|\X_{i})\int_{-1}^{1}(hx+\hat{\X}_i'\Delta(\al))H'(x)dx\cr
	& &+O(1)\int_{-1}^{1}(hx+\hat{\X}_i'\Delta(\al))^{2}|H'(x)|dx\cr
	&=&1-\tau-f(0|\X_{i})\Big{(}\hat{\X}_i'\Delta(\al)+h\int_{-1}^{1}xH^{'}(x)dx\Big{)}\cr
	& &+O(h^{2}+(\hat{\X}_i'\Delta(\al))^{2}).
\end{eqnarray*}
Similarly,
{\small\begin{eqnarray*}
		\ep_{*} \Big{[}\dfrac{\epsilon_i}{h}H'\Big(\dfrac{\epsilon_i-\hat{\X}_i'\Delta(\al)}{h}\Big)\Big{]}&=&h\int_{-1}^{1}(x+\hat{\X}_i'\Delta(\al)/h)H'(x)f(hx+\hat{\X}_i'\Delta(\al)|\X_{i})dx\cr
		&=&f(0|\X_{i})\Big{(}\hat{\X}_i'\Delta(\al)+h\int_{-1}^{1}xH^{'}(x)dx\Big{)}\cr
		& &+O(1)\int_{-1}^{1}(hx+\hat{\X}_i'\Delta(\al))^{2}|H'(x)|dx\cr
		&=&f(0|\X_{i})\Big{(}\hat{\X}_i'\Delta(\al)+h\int_{-1}^{1}xH^{'}(x)dx\Big{)}\cr
		& &+O(h^{2}+(\hat{\X}_i'\Delta(\al))^{2}).
	\end{eqnarray*}
}%
This implies that uniformly in $\al$ and $j$,
\begin{eqnarray}
	|\ep C_{n,h,j}(\al)|\leq C(h^{2}+\|\al-\be(\tau)\|_2^{2}).\label{pr4}
\end{eqnarray}
Hence $\sup_{j}\sup_{k}|\ep C_{n,h,j}(\al_k)|\leq C(h^{2}+a^{2}_{n}).$ Combining that with \eqref{pr1}, \eqref{pr3}, this completes the proof of the proposition under Condition (C3$^{*}$).

We now prove the proposition under condition (C3).  To bound $C_{n,h,j}(\al_{k})-\ep C_{n,h,j}(\al_{k})$ under condition (C3), we introduce the following exponential inequality from \cite{cailiu2011}.
\begin{customlemma}{1 \citep{cailiu2011}}\label{lem:cailiu}
	Let $\xi_1, \dots , \xi_n$ be independent random variables with mean zero. Suppose that
	there exists some $t > 0$ and $\overline{B}_n$ such that $\sum_{k=1}^n\ep \xi_k^2e^{t|\xi_k|}\leq \overline{B}_n^2$. Then for $0<x<\overline{B}_n$,
	\[
	\pr\big(\sum_{k=1}^n \xi_k>C_t\overline{B}_nx\big)\leq\exp(-x^2),
	\]
	where $C_t=t+t^{-1}$.
\end{customlemma}
Note that, by $a_{n}=O(h)$ and (\ref{pr2}),
\begin{eqnarray*}
	|\v_{j}'\X_i H_{h}(\al)|\leq C|\v'_{j}\X_{i}|(1+|(\al-\be(\tau))'\X_{i}|/\|\al-\be(\tau)\|_{2})=:\xi_{i,j}
\end{eqnarray*}
which implies that
\begin{eqnarray*}
	\ep (|\v_{j}'\X_i H_{h}(\al)|)^{2}e^{\eta_{1}|\v_{j}'\X_i H_{h}(\al)|}&\leq& \ep (|\v_{j}'\X_i H_{h}(\al)|)^{2}e^{\eta_{1}\xi_{i,j}}\cr
	&\leq&Ch\ep \Big{[}(\v_{j}'\X_i)^{2}e^{\eta_{1}\xi_{i,j}}(1+|\X'_{i}\Delta(\al)/h|+|\X'_{i}\Delta(\al)/h|^{2})\Big{]}\cr
	&\leq& Ch.
\end{eqnarray*}

Let $\overline{B}_n=C\sqrt{nh}$ and $x=\sqrt{\gamma p\log n}$. By Lemma \ref{lem:cailiu} and the fact that $\sqrt{p\log n}=o(\sqrt{nh})$,  we have
for sufficiently large $C$,
\begin{eqnarray*}
	\sup_{j}\sup_{k}\pr\Big{(}|C_{n,h,j}(\al_{k})-\ep C_{n,h,j}(\al_{k})|\geq C\sqrt{\frac{p h\log n}{n}}\Big{)}=O(n^{-\gamma p}).
\end{eqnarray*}
Combining that with \eqref{pr1} and \eqref{pr4}, we complete the proof of the proposition under condition (C3).
\qed

\subsection{Proof of Proposition \ref{prop2}}
By the proof of Lemma 3 in \cite{cai2010optimal}, we have
\begin{eqnarray*}
	\|\D_{n,h}-\D\|\leq 5\sup_{j\leq b_{p}}|\v'_{j}(\D_{n,h}-\D)\v_{j}|,
\end{eqnarray*}
where $\v_{i}$, $1\leq i\leq b_{p}$, are some non-random vectors with $\|\v_{i}\|_2=1$ and $b_{p}\leq 5^{p}$. Now let
\begin{eqnarray*}
	D_{n,h,j}(\al)=\dfrac{1}{nh}\sum\limits_{i=1}^n (\v'_{j}\X_i)^{2}H'\left(\dfrac{Y_i-\X_i'\al}{h}\right).
\end{eqnarray*}
Therefore when $\|\hat{\be}_{0}-\be(\tau)\|_{2}\leq a_{n}$,
\begin{eqnarray*}
	\sup_{j\leq b_{p}}|\v'_{j}(\D_{n,h}-\D)\v_{j}|\leq \sup_{j\leq b_{p}}\sup_{|\al-\be(\tau)|\leq a_{n}}|\D_{n,h,j}(\al)-\v'_{j}\D \v_{j}|.
\end{eqnarray*}
Note that
\begin{eqnarray*}
	\Big{|}\dfrac{1}{h}H'\left(\dfrac{Y_i-\X_i'\al}{h}\right)-\dfrac{1}{h}H'\left(\dfrac{Y_i-\X_i'\al_{k}}{h}\right)\Big{|}
	\leq Ch^{-2}|\X'(\al-\al_{k})|.
\end{eqnarray*}
Therefore
\begin{eqnarray*}
	& &\sup_{j}\sup_{\|\al-\be(\tau)\|_{2}\leq a_{n}}|\D_{n,h,j}(\al)-\v'_{j}\D \v_{j}|-\sup_{j}\sup_{k}|\D_{n,h,j}(\al_{k})-\v'_{j}\D \v_{j}|\cr
	&\leq& \frac{C\sqrt{p}a_{n}}{n^{M+1}h^{2}}\sum_{i=1}^{n}\|\X_{i}\|_2^{3}.
\end{eqnarray*}
Since $\max_{i,j}\ep|X_{i,j}|^{3}<\infty$, by letting $M$ large enough, we have for any $\gamma>0$,
\begin{eqnarray}
	&&\sup_{j}\sup_{\|\al-\be\|_2\leq a_{n}}|\D_{n,h,j}(\al)-\v'_{j}\D \v_{j}|-\sup_{j}\sup_{k}|\D_{n,h,j}(\al_{k})-\v'_{j}\D \v_{j}|\cr
	\label{eq:d_1} &&\quad=O_{\pr}(n^{-\gamma}).
\end{eqnarray}
We now prove the proposition under under (C3$^*$).
Similarly as in the proof of Proposition \ref{prop1}, define $\hat{X}_{i,j}=X_{i,j}I\{|X_{i,j}|\leq n^{\kappa}\}$ and $\hat{\X}_{i}=(\hat{X}_{i,1},...,\hat{X}_{i,p})'$, and
\begin{eqnarray*}
	\hat{D}_{n,h,j}(\al)=\dfrac{1}{nh}\sum\limits_{i=1}^n (v'_{j}\hat{\X}_i)^{2}H'\Big(\dfrac{\epsilon_i-\hat{\X}_i'\Delta(\al)}{h}\Big).
\end{eqnarray*}
Therefore,
\begin{eqnarray*}
	&&\pr\Big{(}\sup_{j}\sup_{k}|\D_{n,h,j}(\al_{k})-v'_{j}\D v_{j}|\neq \sup_{j}\sup_{k}|\hat{\D}_{n,h,j}(\al_{k})-v'_{j}\D v_{j}|\Big{)}\cr
	&&\quad\leq \sum_{i=1}^{n}\sum_{j=1}^{p}\pr(|X_{i,j}|\geq n^{\kappa})=o(1).
\end{eqnarray*}
As the proof of \eqref{eq:e_h'_sq}, we have
\begin{eqnarray*}
	\sup_{\al}\ep_{*} \Big{[}H'\Big(\dfrac{\epsilon_i-\hat{\X}_i'\Delta(\al)}{h}\Big)\Big{]}^{2}=O(h).
\end{eqnarray*}
By condition (C3$^*$),
\begin{eqnarray*}
	p\log n=o\Big{(}\frac{\sqrt{nhp\log n}}{pn^{2\kappa}}\Big{)}.
\end{eqnarray*}
By Bernstein's inequality, we can show that, for any $\gamma>0$, there exists a constant $C$ such that
\begin{eqnarray}\label{eq:d_bern}
	\qquad\sup_{j}\sup_{k}\pr\Big{(}|\hat{D}_{n,h,j}(\al_{k})-\ep (\hat{D}_{n,h,j}(\al_{k}))|\geq C\sqrt{\frac{p\log n}{nh}}\Big{)}=O(n^{-\gamma p}).
\end{eqnarray}
Moreover,
\begin{eqnarray*}
	\ep_{*} \frac{1}{h}H'\Big(\dfrac{\epsilon_i-\hat{\X}_i'\Delta(\al)}{h}\Big)&=&\int_{-\infty}^{\infty}H'(x)f(hx+\hat{\X}'_{i}\Delta(\al)|\X_{i})dx\cr
	&=&f(0|\X_{i})+O(h+|\hat{\X}'_{i}\Delta(\al)|).
\end{eqnarray*}
Therefore, we have
\begin{eqnarray}
	&&|\ep (\hat{D}_{n,h,j}(\al))-\v'_{j}\D \v_{j}|\cr
	&&\leq C\max_{i} \ep |(\v'_{j}\X_{i})^{2}-(\v'_{j}\hat{\X}_{i})^{2}|+C(h+\|\al-\be(\tau)\|_2)\cr
	&&\leq C\max_{i}(\sum_{j=1}^{p}\ep X^{2}_{i,j}I\{|X_{i,j}|\geq n^{\kappa}\})^{1/2}+C\max_{i}\sum_{j=1}^{p}\ep X^{2}_{i,j}I\{|X_{i,j}|\geq n^{\kappa}\}\cr
	& &+C(h+\|\al-\be(\tau)\|_2)\cr
	&&\leq C(h+\|\al-\be(\tau)\|_2).\label{eq:d_3}
\end{eqnarray}
Combining \eqref{eq:d_3} with \eqref{eq:d_1} and \eqref{eq:d_bern}, we can get the desired inequality under under (C3$^{*}$).

We now prove the proposition under condition (C3).
Let
\begin{eqnarray*}
	\xi_{ij}=(\v'_{j}\X_i)^{2}H'\Big(\dfrac{\epsilon_i-\X_i'\Delta(\al)}{h}\Big).
\end{eqnarray*}
We have
\begin{eqnarray*}
	\ep(\xi_{ij})^{2}e^{\eta_{1}|\xi_{ij}|}&\leq& \ep (\xi_{ij})^{2}e^{C\eta_{1}(\v'_{j}\X_i)^{2}}\cr
	&=&\ep\Big{[}e^{C\eta_{1}(\v'_{j}\X_i)^{2}}(\v'_{j}\X_i)^{2}
	\ep_{*}\Big{[}H'\Big(\dfrac{\epsilon_i-\X_i'\Delta(\al)}{h}\Big)\Big{]}^{2}\Big{]}=O(h).
\end{eqnarray*}
Let $\overline{B}_n=C\sqrt{nh}$ and $x=\sqrt{ \gamma p\log n}$. By Lemma \ref{lem:cailiu} again, we can obtain that
\begin{eqnarray}\label{eq:d_bern2}
	\quad\quad\sup_{j}\sup_{k}\pr\Big{(}|D_{n,h,j}(\al_{k})-\ep (D_{n,h,j}(\al_{k}))|\geq C\sqrt{\frac{p\log n}{nh}}\Big{)}=O(n^{-\gamma p}).
\end{eqnarray}
Combining \eqref{eq:d_1}, \eqref{eq:d_3}, and \eqref{eq:d_bern2}, we can get the desired inequality under under (C3).
\qed

\subsection{Proof of Theorems \ref{thm:beta} and \ref{thm:clt}}

For independent random vectors $\{\X_{i},1\leq i\leq n\}$ with $\sup_{i,j}\ep |X_{ij}|^{3}=O(1)$, let
\begin{eqnarray*}
	\SS_{n}=\sum_{i=1}^{n}\X_{i}(I\{\epsilon_i\geq 0\}+\tau-1).
\end{eqnarray*}
To prove Theorems \ref{thm:beta}-\ref{th4.3}, we need the following lemma.
\begin{lemma}\label{lemma2.3} We have
	\begin{eqnarray}\label{le1}
		\left\|\frac{1}{n}\SS_{n}\right\|_2=O_{\pr}\Big{(}\sqrt{\frac{p}{n}}\Big{)}
	\end{eqnarray}
	and
	\begin{eqnarray}\label{le2}
		\ep\sup_{n\geq 1}\frac{1}{np\log\log n}\|\SS_{n}\|_2^{2}=O(1),
	\end{eqnarray}
	where $O(1)$ is uniformly in $p$.
\end{lemma}

\begin{proof}[Proof of Lemma \ref{lemma2.3}]  We have $\ep \|\SS_{n}\|^{2}=O(np)$. This proves (\ref{le1}). To prove (\ref{le2}), we define
	\begin{eqnarray*}
		S_{nj}=\sum_{i=1}^{n}X_{ij}(I\{\epsilon_i\geq 0\}+\tau-1).
	\end{eqnarray*}
	It is enough to show that
	\begin{eqnarray*}
		\sup_{j}\ep\sup_{n\geq 1}\frac{1}{n\log\log n}|S_{nj}|^{2}=O(1).
	\end{eqnarray*}
	Now,  since $\sup_{i,j}\ep |X_{ij}|^{3}=O(1)$, we have
	\begin{eqnarray*}
		\sum_{i=1}^{\infty}\frac{1}{i\log\log i}\ep X^{2}_{ij}I\{|X_{ij}|\geq \sqrt{i\log\log i}\}\leq C\sum_{i=1}^{\infty} \frac{1}{(i\log\log i)^{3/2}}<\infty.
	\end{eqnarray*}
	Similarly,
	\begin{eqnarray*}
		\sum_{i=1}^{\infty}\frac{1}{\sqrt{i\log\log i}}\ep |X_{ij}|I\{|X_{ij}|\geq \sqrt{i\log\log i}\}\leq C\sum_{i=1}^{\infty} \frac{1}{(i\log\log i)^{3/2}}<\infty.
	\end{eqnarray*}
	The rest proof follows from that of Theorem 1 in \cite{siegmund1969}.
\end{proof}

Now let's prove Theorem \ref{thm:beta} and \ref{thm:clt}. First, by (\ref{le1}) and Propositions \ref{prop1}-\ref{prop2}, it is easy to show that  (\ref{eq:rn00}) holds.

For $q=1$,  note that we assume that $\|\widehat{\be}_{0}-\be(\tau)\|_{2}=O_{\pr}(\sqrt{p/m})$. Then let $a_{n}=\sqrt{p/m}$ and it is easy to see Theorem \ref{thm:beta} holds. Now suppose the theorem holds for $q=g-1$ with some $g\geq 2$.  Noting that $p=O(m/(\log n)^{2})$, we have
$\sqrt{ph_{g-1}(\log n)/n}=O(\sqrt{p/n})$. Then for $q=g$ with initial estimator $\widehat{\be}_{0}=\widehat{\be}^{(g-1)}$, we have $a_{n}=\max(\sqrt{p/n},(p/m)^{2^{ g-2}})=h_{g}$.
By (\ref{eq:rn00}), we have proved Theorem \ref{thm:beta}. Theorem \ref{thm:clt} follows directly from Theorem \ref{thm:beta} and the Lindeberg-Feller central limit theorem. \qed

{
	\subsection{Proof of Propsition \ref{thm:beta'}}
	The proof is identically same as proof of Theorem \ref{thm:beta} with replacing the $\D$ in Proposition \ref{prop2} by $\D=\frac1n\sum_{i=1}^n\ep\X_i\X_i'f_i(0|\X_i)$.
	
	\subsection{Proof of Theorem \ref{th4.3}}
	\begin{proof}
		Define $\lfloor m^{a_{-1}}\rfloor=-m$. Assume that $\|\hat{\be}(\lfloor m^{a_{l-1}}\rfloor)-\be(\tau)\|_{2}=O_{\pr}(\sqrt{p/(m+m^{a_{l-1}})})$.  By Lemma \ref{lem:cailiu} and $p=o(m/(\log m)^{2})$, we have
		uniformly in $k,j$ and $\|\al-\be(\tau)\|_{2}\leq C\sqrt{p/(m+m^{a_{l-1}})}$,
		{\small\begin{eqnarray*}
				\pr\Big{(}\Big{|}\sum_{i=k}^{k+j}v'[\X_i H_{h_{l}}(\al)-\ep \X_i H_{h_{l}}(\al)]\Big{|}\geq C\sqrt{ph_{l}\max(j,m+m^{a_{l-1}})\log m}\Big{)}=O(m^{-\gamma p})
			\end{eqnarray*}
		}%
		for any fixed $l\geq 1$ and uniformly in $\lfloor m^{a_{l-1}}\rfloor+1\leq j\leq \lfloor m^{a_{l}}\rfloor$. Therefore,
		{\small\begin{eqnarray*}
				\Big{\|}\frac{1}{j-\lfloor m^{a_{l-1}}\rfloor }(\U(j)-\V(j)\be(\tau))-\frac{1}{j-\lfloor m^{a_{l-1}}\rfloor}\sum_{i=\lfloor m^{a_{l-1}} \rfloor+1}^{j}\X_{i}\Big{(}I\{\epsilon_{i}\geq 0\}+\tau-1\Big{)}\Big{\|}_{2}\cr
				=O_{\pr}(1)\Big{(}\frac{\sqrt{ph_{l}\max(j,m^{a_{l-1}})\log m}}{j-\lfloor m^{a_{l-1}}\rfloor}+\frac{p}{(m+m^{a_{l-1}})}\Big{)},
			\end{eqnarray*}
		}%
		Therefore, uniformly in $\lfloor m^{a_{l-1}}\rfloor +1\leq j\leq \lfloor m^{a_{l}}\rfloor$,
		\begin{eqnarray*}
			&&\Big{\|}\frac{1}{j-\lfloor m^{a_{l-2}}\rfloor}(\U(\lfloor m^{a_{l-1}}\rfloor)+\U(j)-[\V(\lfloor m^{a_{l-1}}\rfloor)+\V(j)]\be(\tau))\cr
			&&\quad-\frac{1}{j-\lfloor m^{a_{l-2}}\rfloor}\sum_{i=\lfloor m^{a_{l-2}}\rfloor+1}^{j}\X_{i}\Big{(}I\{\epsilon_{i}\geq 0\}+\tau-1\Big{)}\Big{\|}_{2}\cr
			&&\quad=O_{\pr}(1)\Big{(}\sqrt{\frac{ph_{l-1}\log m}{j-\lfloor m^{a_{l-2}}\rfloor}}+\frac{p(\lfloor m^{a_{l-1}} \rfloor-\lfloor m^{a_{l-2}}\rfloor)}{(m+|\lfloor m^{a_{l-2}}\rfloor|)(j-\lfloor m^{a_{l-2}}\rfloor)}+\frac{p}{(m+m^{a_{l-1}})}\Big{)}.
		\end{eqnarray*}
		Similarly, it is easy to show that, uniformly in $\lfloor m^{a_{l-1}} \rfloor+1\leq j\leq \lfloor m^{a_{l}}\rfloor$,
		{\small\begin{eqnarray*}
				&&\Big{\|}\frac{1}{j-\lfloor m^{a_{l-2}}\rfloor}(\V(\lfloor m^{a_{l-1}} \rfloor)+\V(j))-\D\Big{\|}\cr
				&&\quad=O_{\pr}(1)\Big{(}\sqrt{\frac{p\log m}{(j-\lfloor m^{a_{l-2}}\rfloor)h_z{l}}}+\frac{\sqrt{p}(\lfloor m^{a_{l-1}} \rfloor-\lfloor m^{a_{l-2}}\rfloor)}{\sqrt{m+|\lfloor m^{a_{l-2}}\rfloor|}(j-\lfloor m^{a_{l-2}}\rfloor)}+\sqrt{\frac{p}{m+\lfloor m^{a_{l-1}} \rfloor}}\Big{)}.
			\end{eqnarray*}
		}%
		By the second argument in Lemma \ref{lemma2.3}, uniformly in $\lfloor m^{a_{l-1}}\rfloor+1\leq j\leq \lfloor m^{a_{l}}\rfloor$,
		\begin{eqnarray*}
			\frac{\D^{-1}}{j-\lfloor m^{a_{l-2}}\rfloor}\Big{\|}\sum_{i=\lfloor m^{a_{l-2}}\rfloor+1}^{j}\X_{i}\Big{(}I\{\epsilon_{i}\geq 0\}+\tau-1\Big{)}\Big{\|}_{2}=O_{\pr}\Big{(}\sqrt{\frac{p\log m}{j-\lfloor m^{a_{l-2}}\rfloor}}\Big{)}.
		\end{eqnarray*}
		This implies that, uniformly in $\lfloor m^{a_{l-1}}\rfloor+1\leq j\leq \lfloor m^{a_{l}}\rfloor$,
		\begin{eqnarray*}
			&&\Big{\|}\hat{\be}(j)-\be(\tau)-\frac{\D^{-1}}{j-\lfloor m^{a_{l-2}}\rfloor}\sum_{i=\lfloor m^{a_{l-2}}\rfloor+1}^{j}\X_{i}\Big{(}I\{\epsilon_{i}\geq 0\}+\tau-1\Big{)}\Big{\|}_{2}\cr
			&&=O_{\pr}(1)\Big{(}\sqrt{\frac{ph_{l-1} \log m  }{j-\lfloor m^{a_{l-2}}\rfloor}}+\frac{p(\lfloor m^{a_{l-1}}\rfloor-\lfloor m^{a_{l-2}}\rfloor) }{(m+|\lfloor m^{a_{l-2}}\rfloor|)(j-\lfloor m^{a_{l-2}}\rfloor)}+\frac{p \log m }{(m+m^{a_{l-1}})}\Big{)}.
		\end{eqnarray*}
		By the definition of $\{a_{l}\}$, it is easy to show that
		\begin{eqnarray*}
			\frac{\lfloor m^{a_{l-1}}\rfloor-\lfloor m^{a_{l-2}}\rfloor}{(m+|\lfloor m^{a_{l-2}}\rfloor|)(j-\lfloor m^{a_{l-2}}\rfloor)^{1/2}}+\frac{\sqrt{j+m}\log m}{(m+m^{a_{l-1}})}=O(m^{-1/4})
		\end{eqnarray*}
		for $l\geq 1$.
		This completes the proof.
	\end{proof}

	\section{Conjugate gradient method for solving linear systems}
	\label{sec:linear_sys}
	In our simulation study, we use the conjugate gradient method  to efficiently solve the linear system in the construction of LEQR $\widehat{\be}$ in \eqref{eq:beta_0}. To simplify our notation, we rewrite \eqref{eq:beta_0} as $\widehat\be=\V^{-1}\U$, where
	\begin{eqnarray*}
		\V&=&\frac1n\sum\limits_{i=1}^n \X_i\X_i' \dfrac1h  H'\Big(\dfrac{Y_i-\X'_i\widehat{\be}_0}{h}\Big),\cr
		\U&=& \frac1n\sum\limits_{i=1}^n \X_i\Big\{H\Big(\dfrac{Y_i-\X'_i\widehat{\be}_0}{h}\Big)+\tau-1+\dfrac{Y_i}{h}H'\Big(\dfrac{Y_i-\X'_i\widehat{\be}_0}{h}\Big)\Big\}.
	\end{eqnarray*}
	Despite the form of $\widehat\be=\V^{-1}\U$, there is no need to explicitly compute the matrix inversion. Instead, it essentially solves a linear equation system $\V\hat\be=\U$.  
	The \textit{Conjugate gradient method} \citep{hestenes1952methods} is one of the most popular iterative methods for solving a  linear system. 
	For the purpose of completeness, we present the detailed algorithm in Algorithm \ref{algo:cg}. Please refer to Chapter 5 in \cite{Nocedal:06} for more explanations of the algorithm. The time complexity of the conjugate gradient method is $O(s_{\V}\sqrt{\kappa})$ where $s_{\V}$ is the number of non-zero entries of $\V$ and $\kappa$ is its condition number. Since $\V$ is a $p\times p$ matrix so the complexity is at most $O(p^2\sqrt{\kappa})$. In our simulations for the large $p$ case (see Section \ref{sec:large_p} below), we use the R-package \cite{pcg} that implements the conjugate gradient to compute $\hat\be=\V^{-1}\U$ in LEQR. For reference, we test the computation time for solving a linear system using the conjugate gradient method on a standard desktop. A typical run only takes less than $0.05$ seconds when $p=1000$, less than $1$ seconds when $p=5000$, and less than $5$ seconds when $p=10,000$, respectively.
	
	\begin{algorithm}[!t]
		\caption{{\small Conjugate Gradient Method for solving  $\V\be=\U$}}
		\label{algo:cg}
		\hspace{-4.8cm} \textbf{Input}: The matrix $\V\in\R^{p\times p}$ and the vector $\U\in\R^{p}$.
		
		\begin{algorithmic}[1]
			\STATE Initialize $\be_0$. 
			\STATE Set initial residual $\r_0=\U-\V\be_0$.
			\STATE Set initial direction $\boldsymbol{d}_0=\r_0$.
			\STATE Set $k=0$.
			\REPEAT 
			\STATE Update $k=k+1$.
			\STATE Calculate $\alpha_k=\dfrac{\r_{k-1}'\r_{k-1}}{\boldsymbol{d}_{k-1}'\V\boldsymbol{d}_{k-1}}$.
			\STATE Update solution $\be_k=\be_{k-1}+\alpha_k\boldsymbol{d}_{k-1}$.
			\STATE Update residual $\r_k=\r_{k-1}-\alpha_k\V\boldsymbol{d}_{k-1}$.
			\STATE Calculate $\gamma_k=\dfrac{\r_{k}'\r_{k}}{\r_{k-1}'\r_{k-1}}$.
			\STATE Update direction $\boldsymbol{d}_k=\r_k+\gamma_k\boldsymbol{d}_{k-1}$.
			\UNTIL{$\|r_k\|_2$ is  sufficiently small.}
		\end{algorithmic}
		
		\hspace{-8cm} \textbf{Output:}  The final solution $\be_{k}$. 
	\end{algorithm}

	\section{A more general heterogeneous setting}\label{sec:het}
	
	Our algorithm can also be applied to the case that $\{\X_i,\epsilon_i\}$'s are independent, but not identically distributed from  model \eqref{eq:linear_model}.  For example, in sensor networks, sensors are placed at the different locations. Therefore, the covariates $\X_i$'s (e.g., climate factors) can be differently distributed across different sensors.  Note that the underlying true coefficient vector $\be(\tau)$ must be identical across different machines since otherwise the distributed inference problem  is ill-defined.
	In order to obtain the Bahadur representation in Theorem \ref{thm:beta}, we need to revise the conditions (C1), (C2), (C3) and (C3$^*$) as the following,
	\begin{enumerate}
		\item[(HC1)] For each data $1\leq i \leq n$, the conditional density function $f_i(\cdot|\X_i)$ is Lipschitz continuous with the Lipschitz constant $L_i$ bounded uniformly by some positive constant $L$ for all $i$. Let $\D_n=\frac1n\sum_{i=1}^n\ep[\X_i\X_i'f_i(0|\X_i)]$. We assume that $0<c_{1}\leq \lambda_{\min}(\D_n)\leq\lambda_{\max}(\D_n)\leq c_{2}<\infty$ for some constants $c_{1},c_{2}$  uniformly in $n$.
		
		\item[(HC2)] The same as (C2).
		\item[(HC3)] Assume that $p=o\left(nh/(\log n)\right)$ and $\sup_i\sup_{\|\thh\|_{2}=1}\ep e^{\eta(\thh'\X_i)^{2}}<\infty$ for some $\eta>0$.
		
		\item[(HC3$^{*}$)]  Assume for some $\kappa>0$, $p=o\left((n^{1-4\kappa}h/\log n)^{1/3}\right)$. Suppose that $\sup_{i,j}\ep |X_{i,j}|^{a}<\infty$
		for some $a\geq 2/\kappa$ and $\sup_i\sup_{\|\thh\|_2=1}\ep (\thh'\X_i)^{4}<\infty$.
	\end{enumerate}
	The conditions (HC1), (HC3) and (HC3$^{*}$) are similar to those in Section \ref{sec:DC_leqr} except that the conditions on covariates $\X$ are changed to uniform bounds. Under these conditions, we present the Bahadur representation of $\hat{\be}^{(q)}$.

	
	\begin{proposition}\label{thm:beta'}
		Assume the initial estimator $\hat{\be}_{0}$  in (\ref{ag0})  satisfies $\|\hat{\be}_{0}-\be(\tau)\|_2=O_{\pr}(\sqrt{p/m})$. Let
		$h_{g}=\max(\sqrt{p/n},(p/m)^{2^{g-2}})$ for $1\leq g\leq q$. Assuming that conditions (HC1), (HC2), and (HC3) (or (HC3$^*$)) hold  with $h=h_q$ and $p=O(m/(\log n)^2)$, we have that,
		\begin{eqnarray*}
			\hat{\be}^{(q)}-\be(\tau)=\frac{\D^{-1}}{n}\sum_{i=1}^{n}\X_{i}(I\{\epsilon_{i}\geq 0\}+\tau-1)+\r_{n}
		\end{eqnarray*}
		with
		\begin{eqnarray*}
			\|\r_{n}\|_2=O_{\pr}\Big{(}\sqrt{\frac{ph_{q}\log n}{n}}+\Big{(}\frac{p}{m}\Big{)}^{2^{q-1}}\Big{)}.
		\end{eqnarray*}
	\end{proposition}
	The proof of Proposition \ref{thm:beta'} is quite similar to the proof of Theorem \ref{thm:beta} (since the proofs of our Propositions \ref{prop1}-\ref{prop2} do not use identically distributed property), and we mainly need to replace $\D$ with $\frac1n\sum_{i=1}^n\ep[\X_i\X_i'f_i(0|\X_i)]$. Furthermore, a similar central limit theorem holds as in Theorem \ref{thm:clt}.
	

	
	\section{Limitation of  the na\"ive-DC approach}
	\label{sec:naiveDC}
	
	In this section, we demonstrate the limitation of the na\"ive-DC approach. In particular, on a special case of quantile regression problem with $p=0$ (i.e., inference for quantile), we show that the na\"ive-DC estimator fails when $n \geq cm^2$.

	In particular,  let $Y_1, \ldots, Y_n$ be $n$ \emph{i.i.d.} samples from a population $Y$, which has the density function $f(y)$. Further, let $\beta(\tau)$ be the $\tau$-th quantile of $Y$, i.e., $\pr(Y\leq \beta(\tau))=\tau$. Denote by $\widehat{\beta}_k$ the $\tau$-th  sample quantile computed on $\mathcal{D}_k$. There are different ways in the literature to define the $\tau$-th  sample quantiles (see \cite{hyndman1996sample}), while most of them have the following form
	\[
	\widehat{\beta}_k= (1 - \gamma) Y_{\mathcal{H}_k,(j)} + \gamma\cdot Y_{\mathcal{H}_k,(j+1)},
	\]
	where index $j$ and $\gamma$ depend on $\tau$. Here $Y_{\mathcal{H}_k,(j)}, j=1,2,\dots,m$ are the order statistics of $\mathcal{D}_k=\{Y_i\}_{i\in \mathcal{H}_k}$. We use a popular choice of $j$ and $\gamma$ that $j=\lfloor\tau(m+1)\rfloor$ and $\gamma=\tau(m+1)-j$.
	
	Denote by $\widehat{\beta}=\frac{1}{N}\sum_{k=1}^{N}\widehat{\beta}_{k}$ the na\"ive-DC estimator of $\tau$-th quantile.  We have the following theorem for the asymptotic distribution of $\widehat{\beta}$.
	
	\begin{theorem} \label{thm:sample_quantile}
		Suppose the density function $f(\cdot)$ is positive and has bounded third derivatives. We have, for $n\leq m^{A}$ with some constant $A>0$,
		\begin{eqnarray*}
			\widehat{\beta}-\beta(\tau)&=&\frac{1}{nf(\beta(\tau))}\sum_{i=1}^{n}\Big{(}\tau-I\{Y_{i}\leq \beta(\tau)\}\Big{)}-\frac{\tau(1-\tau)f'(\beta(\tau))}{2mf^{3}(\beta(\tau))} \cr &&+O_{\pr}\Big{(}\frac{\log m}{\sqrt{n}m^{1/4}}+\dfrac{(\log m)^{2}}{m^{5/4}}\Big{)}.
		\end{eqnarray*}
		In particular:
		
		\begin{enumerate}
			\item[(1)] If $n/m^2\rightarrow 0$, then
			\begin{eqnarray*}
				\sqrt{n}(\widehat{\beta}-\beta(\tau))\Rightarrow N\Big(0,\dfrac{\tau(1-\tau)}{[f(\beta(\tau))]^2}\Big),\quad\text{as}\quad n\rightarrow \infty.
			\end{eqnarray*}
			\item[(2)] If $n/m^2\rightarrow c$ for some $c>0$, then
			\begin{eqnarray*}
				\sqrt{n}(\widehat{\beta}-\beta(\tau))-\sqrt{c}b\Rightarrow N\Big(0,\dfrac{\tau(1-\tau)}{[f(\beta(\tau))]^2}\Big),\quad\text{as}\quad n\rightarrow \infty,
			\end{eqnarray*}
			where $b=-\frac{\tau(1-\tau)f'(\beta(\tau))}{2f^{3}(\beta(\tau))}$.
			\item[(3)] If $n/m^{2}\rightarrow \infty$, $n\leq m^{A}$ for some $A>0$, and $b\neq 0$,  then we have
			\begin{eqnarray*}
				\frac{m(\widehat{\beta}-\beta(\tau))}{b}\rightarrow 1
			\end{eqnarray*}
			in probability.
		\end{enumerate}
	\end{theorem}
	
	Theorem \ref{thm:sample_quantile} shows that when $n/m^{2} \rightarrow \infty$, the estimation accuracy of $\widehat{\beta}$ only depends on $m$ and  increasing the sample size $n$ in the na\"ive-DC method does not improve the estimation accuracy. Also, applying the na\"ive-DC method to streaming data requires $\Omega(\sqrt{n})$ space to have the asymptotic normality and optimal convergence rate.
	
	\subsection{Proof of Theorem \ref{thm:sample_quantile}}
	\label{sec:proof_sample_quantile}
	Let $F$ and $F_k$ be the distribution of $Y$ and the empirical distribution on $\mathcal{D}_k$.
	It follows directly from  the proof in Bahadur (1966) that, uniformly for $1\leq k\leq N$,
	\begin{equation}\label{eq:bahadur}
		Y_{\mathcal{H}_{k},(j)}-\beta(\tau)=\dfrac{1}{mf(\beta(\tau))}\sum\limits_{i\in \mathcal{H}_k}\Big(\tau-I\{Y_i\leq \beta(\tau)\}\Big)+O_{\pr}\left(\frac{\log m}{m^{3/4}}\right),
	\end{equation}
	where $j=\lfloor{\tau(m+1)}\rfloor$,
	which in turn implies that $\max_{1\leq k\leq N}|\hat{\beta}_{k}-\beta(\tau)|=O_{\pr}(\sqrt{(\log m)/m})$.
	
	For the simplicity of the proof, we denote by $\beta=\beta(\tau)$.
	Since $F$ has bounded third derivatives, for any $k$,
	\[
	F(Y_{\mathcal{H}_{k},(j)})-F(\beta)=f(\beta)(Y_{\mathcal{H}_{k},(j)}-\beta)+\frac{1}{2}f'(\beta)(Y_{\mathcal{H}_{k},(j)}-\beta)^2+R,
	\]
	where $R\leq C|Y_{\mathcal{H}_{k},(j)}-\beta|^3$ for some constant $C$.
	Hence
	\begin{equation}\label{eq:sq_err}
		Y_{\mathcal{H}_{k},(j)}-\beta=\dfrac{F(Y_{\mathcal{H}_{k},(j)})-F(\beta)}{f(\beta)}-\frac12\dfrac{f'(\beta)}{f(\beta)}(Y_{\mathcal{H}_{k},(j)}-\beta)^2-\dfrac{R}{f(\beta)}.
	\end{equation}
	Denote by $C_1^{(k)}=F(Y_{\mathcal{H}_{k},(j)})-F(\beta)-\left[\tau-F_k(\beta)\right]$.
	By the proof of Lemma 1 in Bahadur (1966), we have, for any $c_{1}>0$, there exists a sufficiently large $c_{1}>0$ such that
	\begin{eqnarray*}
		\pr\Big{(}|C_1^{(k)}|\geq c_{1}\dfrac{\log m}{m^{3/4}}\Big{)}=O(m^{-c_{2}}).
	\end{eqnarray*}
	This implies that $\ep (A^{k}_{1})^{2}=O(\frac{(\log m)^{2}}{m^{3/2}})$.
	Note that $F(x)$ is an nondecreasing function and $F(Y)$ is uniform in $[0,1]$.
	We have $\ep F(Y_{\mathcal{H}_{k},(j)})=j/(m+1)$. Therefore,
	\begin{eqnarray*}
		\frac{1}{N}\sum_{k=1}^{N}C_1^{(k)}&=&\frac{1}{N}\sum_{k=1}^{N}(C_1^{(k)}-\ep C_1^{(k)})+j/(m+1)-\tau\cr
		&=&O_{\pr}\Big{(}\frac{\log m}{\sqrt{n}m^{1/4}}\Big{)}+j/(m+1)-\tau.
	\end{eqnarray*}
	By (\ref{eq:bahadur}), it is easy to see
	\[
	(Y_{\mathcal{H}_{k},(j)}-\beta)^2=\dfrac{1}{m^{2}f^{2}(\beta)}\Big{[}\sum\limits_{i\in \mathcal{H}_k}\Big(\tau-I\{Y_i\leq \beta\}\Big)\Big{]}^{2}+O_{\mathrm{p}}\left(\dfrac{(\log m)^{2}}{m^{5/4}}\right)
	\]
	uniformly for $1\leq k\leq N$. Therefore
	\begin{eqnarray*}
		\frac{1}{N}\sum_{k=1}^{N}(Y_{\mathcal{H}_{k},(j)}-\beta)^2=\frac{\tau(1-\tau)}{mf^{2}(\beta(\tau))}+O_{\mathrm{p}}\left(\dfrac{(\log m)^{2}}{m^{5/4}}\right).
	\end{eqnarray*}
	Combining the above arguments, we obtain
	\begin{eqnarray*}
		Y_{\mathcal{H}_{k},(j)}-\beta&=&\frac{1}{n}\sum_{i=1}^{n}(\tau-I\{Y_{i}\leq \beta\})-\frac{\tau(1-\tau)f'(\beta)}{2mf^{3}(\beta)}+j/(m+1)-\tau\cr
		& &+O_{\pr}\Big{(}\frac{\log m}{\sqrt{n}m^{1/4}}+\dfrac{(\log m)^{2}}{m^{5/4}}\Big{)}.
	\end{eqnarray*}
	It is easy to see that the above equation also holds for $Y_{\mathcal{H}_{k},(j+1)}$. The proof is complete.
	
	\section{Additional simulations}
	\label{sec:supp_exp}

	
	
	\subsection{Coverage rates, bias and variance analysis when $p=3$}
	\label{sec:p_15}
	In Tables \ref{table:bias_4}--\ref{table:bias_6}, we report the performance of our proposed DC LEQR, na\"ive-DC, and QR All when $p=3$. The settings of batch size $m$, the total sample size $n$, the quantile level $\tau$, and the distribution of the noise remain the same as in Tables \ref{table:bias_1}--\ref{table:bias_3}. From \eqref{eq:q}, it is easy to see that now we need the number of aggregations $q\geq3$. Consequently we report the performance of our DC LEQR when $q=3$ or $q=4$. As one can see, na\"ive-DC still fails when $n$ is large,  while the coverage rates of our DC LEQR remain at the nominal level (or similar to QR All) when  $q=4$.
	
	\subsection{Sensitivity analysis for other noise types}
	\label{sec:supp_sens}
	In this section, we reproduce Tables \ref{table:adaptive}--\ref{table:sesd_adaptive} with the other two noise types (e.g., heteroscedastic normal and exponential). The results are reported in Tables \ref{table:adaptive_2}--\ref{table:sesd_adaptive_3}.  Similar to the case of homoscedastic normal, as long as the number of iterations $q \geq 3$, all different choices of scaling constants lead to coverage rates close to the nominal level of 95\%.
	
	
	\subsection{DC LEQR for large $p$}
	\label{sec:large_p}
	We generate the data from the model in \eqref{eq:simu_model} and adopt the same setting described in Section \ref{sec:simu}. We consider two types of noises: homoscedastic normal, $\epsilon_i\sim N(0,1)$; and heteroscedastic normal, $\epsilon_i\sim N(0,(1+0.3X_{i1})^2)$. The dimension $p$ is set to $1000$. We vary the sample size $n$ from $10^6$ to $2\times 10^6$ and the batch size $m$ from $10,000$ to $20,000$. We set the nominal level to be $95\%$ and the quantile level $\tau = 0.1$.
	In Tables \ref{table:large_1}--\ref{table:large_2}, we report the coverage rates, bias and variance of the proposed DC LEQR with number of iterations $q =4$ and $q=8$. For comparison, we also report the performance of QR All (i.e., the QR estimator in \eqref{eq:QR} computed on all data points) and na\"ive-DC. All results reported are the average of 1000 independent runs.
	
	From the results in Tables \ref{table:large_1}--\ref{table:large_2}, the coverage rates of the proposed DC LEQR are close to the nominal level 95\% after $q=4$ iterations. In contrast, the coverage rates of the na\"ive-DC estimator  are lower than the nominal level when the number of machines is large (i.e., either the batch size $m$ is small for a fixed sample size $n$ or the sample size $n$ is large for a fixed batch size $m$). Similar to the bias and variance analysis in Tables \ref{table:bias_1}--\ref{table:bias_3}, the bias of na\"ive-DC is much larger than the proposed DC LEQR and QR All, which explains the reason of the failure of na\"ive-DC in terms of  inference. The variance of the proposed DC LEQR decreases when the number of iterations $q$ increases. When $q=8$, the variance of DC LEQR is very close to that of the na\"ive-DC as well as QR All. In summary, when the dimension $p$ is large, the proposed DC LEQR algorithm still achieves desirable performance with a small number of iterations.


	\subsection{Comparison with ADMM}\label{sec:admm}
	
	Recently, researchers have developed new optimization techniques based on  alternating direction method of multiplier (ADMM)  for solving QR problems (see, e.g., \cite{yu2017parallel,gu2017admm}). We refer the readers to \cite{boyd2011distributed} for more details on ADMM. While their work mostly focuses on optimizing a high-dimensional QR objective with penalization, the proposed method can also be modified to optimizing un-regularized QR problems.
	
	We provide a comparison between DC LEQR and ADMM-based QR estimator in Table \ref{table:ADMM}. We conduct experiments for different $m$ and $n$ with $\tau=0.1$, $p=15$, and $\epsilon_i\sim N(0,1)$. We report empirical coverage rates, mean bias and variance of $\v_0'\be$ and computation time as an average of 1000 independent runs of the simulations.
	
	Due to our limited computing resource, we implement the both methods on a single machine but record the time as in a parallel setting. In particular, we separate the data to $N$ subsets. For DC LEQR, we first compute the initial estimator and record the computation time $T_0$. Next, for each round $g=1,2,\dots,q$, we compute the local statistic for each batch of data $k$ for $k=1,2,\dots, N$ and record the local computation time $T_{1,k}^{(g)}$. We take the maximum of the local computation time for different batches of data as an estimated time for this step (i.e., $\sum_{g=1}^q\max_kT_{1,k}^{(g)}$).      Then, we aggregate the local statistics by solving a linear system in \eqref{ag1} and record the aggregation time $T_2^{(g)}$. The total computation time for LEQR is calculated as $T_0+\sum_{g=1}^q\max_kT_{1,k}^{(g)}+\sum_{g=1}^qT_2^{(g)}$. For ADMM, denote by $\r=\y-\X\be$ where $\y=(y_1,y_2,\dots, y_n)'$ and $\X=(\X_1,\X_2,\dots, \X_n)'$. The primal update of $\be$ can be divided into two steps that fit a parallel computing scheme as suggested in \cite{yu2017parallel}. Following \cite{yu2017parallel}, in the $g$-th iteration, each of the primal $\be$-step, $\r$-step, and the dual step can all be updated separately for each batch of data by allowing the communication of the estimator of the $(g-1)$-th iteration $\widehat{\be}^{(g-1)}$. For each iteration $g$ and data batch $k$, we record the local computation time $\widetilde{T}_{1,k}^{(g)}$ and aggregation time $\widetilde{T}_2^{(g)}$. The total computation time for ADMM is calculated as $\sum_{g=1}^{q_1}\max_k \widetilde{T}_{1,k}^{(g)}+\sum_{g=1}^{q_1} \widetilde{T}_2^{(g)}$ where $q_1$ is the number of iterations.
	
	The recorded computation time of the both methods is analogous to the distributed setting despite that the communication time is neglected.  On the other hand, the communication time could vary a lot from one setting to another depending on the communication protocol (e.g., communication within a single machine, within a local area network, or on the internet). Therefore, we will investigate the theoretical communication cost in the next paragraph. Minimizing the communication time in practice under different communication protocols is not the main focus of the paper and we would like to leave it for future investigation.
	

	In terms of total communication cost, our DC LEQR requires $O(qp^2(n/m))$ bits, where $q$ is the number of iterations and $n/m$ is the number of data batches. Correspondingly, ADMM requires $O(q_1p(n/m))$ bits, where $q_1$ is the number of iterations of ADMM, which depends on the condition number of the Hessian matrix (and may depend on $p$).  Although the number of iterations $q$ has an explicit formula with dimension $p$ (see Eq. \eqref{eq:q}),  we can not find a precise relation between $q_{1}$ and $p$ in the existing literature on quantile regression. Hence it is a bit hard to compare the communication cost rigorously. In our simulation setup with $p=15$, the ADMM algorithm often stops within around 50 to 100 iterations. For example, a typical trail of ADMM stops at $84$, $88$, $55$, $57$ iterations when $(m,n)=(100, 10^6)$, $(500,10^6)$, $(500,10^6)$, $(1000,10^6)$, respectively.


	As one can see from Table \ref{table:ADMM}, on one hand, since ADMM optimizes the QR objective on all the samples, the coverage rates of the obtained solution are close to the nominal level of 95\%. On the other hand, our method is still  faster than ADMM-based method.

	\subsection{Online LEQR}
	\label{sec:online_exp}
	In addition, in Figure \ref{fig:online}, we investigate the performance of online LEQR. We fix the memory size $m=500$ and choose the dimensionality from $p\in\{3,15\}$.  We consider the homoscedastic normal noises as described in the beginning of Section \ref{sec:simu}. The size of online streaming data $n$ (i.e., $m+j$ in Theorem \ref{th4.3}) varies from $n=10^4$ to $n=10^6$. From Theorem \ref{th4.3}, an oracle $(1- \alpha_0)$-th confidence interval for $\v_0'\be(\tau)$ is given by
	\begin{equation}
		\v_0'\widehat{\be}[n-m]\pm n^{-1/2}\sqrt{\tau(1-\tau)\v_0'\D^{-1}\ep\left[\vec{X}\vec{X}'\right]\D^{-1}\v_0}z_{\alpha_0/2},
	\end{equation}
	where $z_{\alpha_0/2}$ is the $(1-\alpha_0/2)$-quantile of the standard normal distribution. Similarly as the DC LEQR, $\D$ and $\ep\left[\vec{X}\vec{X}'\right]$ are estimated by $\frac{1}{n-r_{l-2}}(\V(r_{l-1})+\V(n))$ and $\frac{1}{n} \sum_{i=1}^n \X_i \X_i'$, respectively.  We report the empirical coverage rate as an average of 1000 independent runs of the simulations.   For all the settings, the coverage rates are close to the nominal level $95\%$ when the sample size $n$ is reasonably large (e.g., $n \geq 10^5$). Note that for many emerging internet applications involving online streaming data, the data sequence is long.
	
	\begin{table}[!t]
		\centering
		\caption{Coverage rates, bias, and variance of DC LEQR v.s. QR All (i.e., the quantile regression estimator in \eqref{eq:QR} computed on all data points) and na\"ive-DC when $n$ varies from $m^{1.6}$ to $m^{3}$. Noises $\epsilon_i$'s are generated from homoscedastic normal distribution. Dimension $p=3$. Batch size $m=100$. Quantile level $\tau \in \{0.1,0.5,0.9\}$.}
		\label{table:bias_4}\label{table:coverage}
		\begin{tabular}{ll|crr|crr}
			\hline
			&$\log_m(n)$&Coverage&Bias ~~~  & Var ~~~~~ &Coverage&Bias~~~  & Var ~~~~~\\
			&&Rate&($\times 10^{-2}$) & ($\times 10^{-4}$)&Rate&($\times 10^{-2}$) & ($\times 10^{-4}$)\\ \hline
			&          & \multicolumn{3}{c|}{DC LEQR $q=3$} &   \multicolumn{3}{c}{DC LEQR $q=4$} \\
			\hline
			\multicolumn{2}{l|}{$\tau=0.1$}  &&&&&& \\
			& 1.6 & 0.932  & 0.21  & 32.32 & 0.935    & 0.07  & 31.46 \\
			& 2.0   & 0.948  & 0.14  & 4.47  & 0.950    & 0.14  & 4.44  \\
			& 2.4 & 0.964  & 0.01  & 0.65  & 0.967    & 0.00  & 0.63  \\
			& 3.0   & 0.939  & 0.02  & 0.05  & 0.958    & 0.00  & 0.04  \\
			\multicolumn{2}{l|}{$\tau=0.5$}  &&&&&& \\
			& 1.6 & 0.941  & 0.07  & 15.60 & 0.942    & 0.07  & 15.63 \\
			& 2.0   & 0.946  & -0.05 & 2.25  & 0.945    & -0.05 & 2.25  \\
			& 2.4 & 0.955  & -0.01 & 0.38  & 0.955    & -0.01 & 0.38  \\
			& 3.0   & 0.947  & -0.01 & 0.02  & 0.947    & -0.01 & 0.02  \\
			\multicolumn{2}{l|}{$\tau=0.9$} &&&&&& \\
			& 1.6 & 0.938  & -0.38 & 30.70 & 0.936    & -0.26 & 30.53 \\
			& 2.0   & 0.939  & 0.03  & 4.55  & 0.943    & 0.04  & 4.49  \\
			& 2.4 & 0.948  & 0.01  & 0.71  & 0.949    & 0.02  & 0.69  \\
			& 3.0   & 0.924  & -0.03 & 0.07  & 0.956    & -0.01 & 0.04  \\
			\hline
			&&\multicolumn{3}{c|}{QR All} &\multicolumn{3}{c}{Na\"ive-DC}     \\
			\hline
			\multicolumn{2}{l|}{$\tau=0.1$}  &&&&&& \\
			& 1.6 & 0.944  & 0.06  & 29.80 & 0.933    & 1.20  & 30.35 \\
			& 2.0   & 0.940  & -0.05 & 4.48  & 0.913    & 1.19  & 4.39  \\
			& 2.4 & 0.951  & 0.02  & 0.71  & 0.687    & 1.20  & 0.71  \\
			& 3.0   & 0.948  & -0.01 & 0.04  & 0.000    & 1.20  & 0.04  \\
			\multicolumn{2}{l|}{$\tau=0.5$}  &&&&&& \\
			& 1.6 & 0.939  & -0.04 & 16.44 & 0.947    & -0.10 & 16.91 \\
			& 2.0   & 0.960  & 0.01  & 2.42  & 0.954    & 0.01  & 2.43  \\
			& 2.4 & 0.943  & -0.03 & 0.39  & 0.933    & -0.01 & 0.39  \\
			& 3.0   & 0.937  & 0.01  & 0.03  & 0.946    & 0.01  & 0.03  \\
			\multicolumn{2}{l|}{$\tau=0.9$}  &&&&&& \\
			& 1.6 & 0.952  & -0.17 & 30.36 & 0.948    & -0.83 & 30.39 \\
			& 2.0   & 0.958  & -0.01 & 4.06  & 0.938    & -0.69 & 4.24  \\
			& 2.4 & 0.955  & -0.05 & 0.72  & 0.871    & -0.71 & 0.67  \\
			& 3.0   & 0.940  & 0.00  & 0.05  & 0.077    & -0.70 & 0.04 \\
			\hline
		\end{tabular}
	\end{table}

	\begin{table}[!t]
		\centering
		\caption{Coverage rates, bias, and variance of DC LEQR v.s. QR All (i.e., the QR estimator in \eqref{eq:QR} computed on all data points) and na\"ive-DC when $n$ varies from $m^{1.6}$ to $m^{3}$. Noises $\epsilon_i$'s are generated from heteroscedastic normal distribution. Dimension $p=3$. Batch size $m=100$. Quantile level $\tau \in \{0.1,0.5,0.9\}$. 
		}
		\label{table:bias_5}
		\begin{tabular}{ll|crr|crr}
			\hline
			&$\log_m(n)$&Coverage&Bias ~~~  & Var ~~~~~ &Coverage&Bias~~~  & Var ~~~~~\\
			&&Rate&($\times 10^{-2}$) & ($\times 10^{-4}$)&Rate&($\times 10^{-2}$) & ($\times 10^{-4}$)\\ \hline
			&          & \multicolumn{3}{c|}{DC LEQR $q=3$} &   \multicolumn{3}{c}{DC LEQR $q=4$} \\
			\hline
			\multicolumn{2}{l|}{$\tau=0.1$}  &&&&&& \\
			& 1.6 & 0.940  & 0.54  & 44.18 & 0.938    & 0.34  & 43.47 \\
			& 2.0 & 0.957  & 0.08  & 6.37  & 0.955    & 0.07  & 6.33  \\
			& 2.4 & 0.945  & 0.01  & 1.10  & 0.948    & 0.00  & 1.02  \\
			& 3.0 & 0.942  & 0.02  & 0.08  & 0.959    & 0.01  & 0.06  \\
			\multicolumn{2}{l|}{$\tau=0.5$}  &&&&&& \\
			& 1.6 & 0.951  & 0.12  & 21.01 & 0.950    & 0.12  & 20.99 \\
			& 2.0 & 0.958  & 0.01  & 3.06  & 0.956    & 0.01  & 3.06  \\
			& 2.4 & 0.955  & -0.02 & 0.54  & 0.955    & -0.02 & 0.54  \\
			& 3.0 & 0.954  & 0.01  & 0.03  & 0.955    & 0.01  & 0.03  \\
			
			\multicolumn{2}{l|}{$\tau=0.9$} &&&&&&  \\
			& 1.6 & 0.937  & -0.33 & 41.89 & 0.939    & -0.11 & 42.03 \\
			& 2.0 & 0.944  & -0.01 & 6.32  & 0.946    & 0.01  & 6.13  \\
			& 2.4 & 0.950  & 0.00  & 1.05  & 0.956    & 0.01  & 0.97  \\
			& 3.0 & 0.934  & -0.03 & 0.09  & 0.951    & -0.01 & 0.06  \\
			\hline
			&&\multicolumn{3}{c|}{QR All} &\multicolumn{3}{c}{Na\"ive-DC}     \\
			\hline
			\multicolumn{2}{l|}{$\tau=0.1$}  &&&&&& \\
			& 1.6 & 0.937  & 0.27  & 45.19 & 0.921    & 2.40  & 40.71 \\
			& 2.0 & 0.953  & 0.15  & 6.44  & 0.834    & 2.47  & 6.21  \\
			& 2.4 & 0.952  & 0.03  & 0.96  & 0.345    & 2.37  & 0.93  \\
			& 3.0 & 0.938  & 0.00  & 0.06  & 0.000    & 2.34  & 0.06  \\
			\multicolumn{2}{l|}{$\tau=0.5$}  &&&&&& \\
			& 1.6 & 0.957  & -0.03 & 22.72 & 0.941    & 0.04  & 23.23 \\
			& 2.0 & 0.945  & -0.01 & 3.59  & 0.946    & -0.02 & 3.58  \\
			& 2.4 & 0.946  & 0.02  & 0.53  & 0.943    & 0.01  & 0.54  \\
			& 3.0 & 0.946  & 0.00  & 0.03  & 0.944    & 0.00  & 0.03  \\
			\multicolumn{2}{l|}{$\tau=0.9$}  &&&&&& \\
			& 1.6 & 0.947  & -0.22 & 42.49 & 0.961    & -1.46 & 39.35 \\
			& 2.0 & 0.934  & -0.04 & 6.28  & 0.911    & -1.37 & 6.46  \\
			& 2.4 & 0.960  & 0.02  & 0.90  & 0.752    & -1.30 & 0.95  \\
			& 3.0 & 0.957  & 0.01  & 0.06  & 0.001    & -1.32 & 0.06  \\
			\hline
		\end{tabular}
	\end{table}
	
	\begin{table}[!t]
		\centering
		\caption{Coverage rates, bias, and variance of DC LEQR v.s. QR All (i.e., the QR estimator in \eqref{eq:QR} computed on all data points) and na\"ive-DC when $n$ varies from $m^{1.6}$ to $m^{3}$. Noises $\epsilon_i$'s are generated from exponential distribution. Dimension $p=3$. Batch size $m=100$. Quantile level $\tau \in \{0.1,0.5,0.9\}$.}
		\label{table:bias_6}
		\begin{tabular}{ll|crr|crr}
			\hline
			&$\log_m(n)$&Coverage&Bias ~~~  & Var ~~~~~ &Coverage&Bias~~~  & Var ~~~~~\\
			&&Rate&($\times 10^{-2}$) & ($\times 10^{-4}$)&Rate&($\times 10^{-2}$) & ($\times 10^{-4}$)\\ \hline
			&          & \multicolumn{3}{c|}{DC LEQR $q=3$} &   \multicolumn{3}{c}{DC LEQR $q=4$} \\
			\hline
			\multicolumn{2}{l|}{$\tau=0.1$}  &&&&&& \\
			& 1.6 & 0.951  & 0.24  & 0.86  & 0.950    & 0.15  & 0.96  \\
			& 2.0 & 0.958  & -0.01 & 0.17  & 0.956    & -0.01 & 0.17  \\
			& 2.4 & 0.955  & -0.01 & 0.03  & 0.944    & 0.00  & 0.03  \\
			& 3.0 & 0.954  & 0.00  & 0.00  & 0.941    & 0.00  & 0.00  \\
			\multicolumn{2}{l|}{$\tau=0.5$}  &&&&&& \\
			& 1.6 & 0.941  & -0.09 & 9.80  & 0.955    & -0.06 & 9.84  \\
			& 2.0 & 0.936  & 0.02  & 1.64  & 0.943    & 0.02  & 1.64  \\
			& 2.4 & 0.949  & -0.01 & 0.24  & 0.944    & -0.01 & 0.24  \\
			& 3.0 & 0.950  & 0.00  & 0.02  & 0.953    & 0.00  & 0.02  \\
			\multicolumn{2}{l|}{$\tau=0.9$}  &&&&&& \\
			& 1.6 & 0.941  & -0.62 & 92.70 & 0.940    & -0.25 & 85.62 \\
			& 2.0 & 0.936  & 0.04  & 14.32 & 0.937    & 0.11  & 13.44 \\
			& 2.4 & 0.949  & -0.06 & 3.15  & 0.949    & -0.01 & 2.30  \\
			& 3.0 & 0.950  & -0.18 & 2.24  & 0.953    & -0.01 & 0.30  \\
			\hline
			&&\multicolumn{3}{c|}{QR All} &\multicolumn{3}{c}{Na\"ive-DC}     \\
			\hline
			\multicolumn{2}{l|}{$\tau=0.1$}  &&&&&& \\
			& 1.6 & 0.947  & 0.09  & 1.10  & 0.835    & 1.00  & 1.28  \\
			& 2.0 & 0.956  & 0.02  & 0.16  & 0.324    & 1.00  & 0.19  \\
			& 2.4 & 0.941  & 0.00  & 0.03  & 0.001    & 0.98  & 0.03  \\
			& 3.0 & 0.951  & 0.00  & 0.00  & 0.000    & 0.99  & 0.00  \\
			\multicolumn{2}{l|}{$\tau=0.5$}  &&&&&& \\
			& 1.6 & 0.937  & 0.14  & 11.10 & 0.930    & 1.12  & 11.00 \\
			& 2.0 & 0.952  & 0.04  & 1.45  & 0.856    & 1.04  & 1.59  \\
			& 2.4 & 0.950  & -0.01 & 0.23  & 0.463    & 1.01  & 0.23  \\
			& 3.0 & 0.955  & 0.00  & 0.02  & 0.000    & 1.03  & 0.02  \\
			\multicolumn{2}{l|}{$\tau=0.9$}  &&&&&& \\
			& 1.6 & 0.953  & -0.27 & 86.11 & 0.945    & 1.79  & 83.42 \\
			& 2.0 & 0.958  & -0.06 & 12.91 & 0.938    & 1.68  & 12.05 \\
			& 2.4 & 0.946  & -0.04 & 2.22  & 0.766    & 1.75  & 2.12  \\
			& 3.0 & 0.958  & -0.01 & 0.14  & 0.002    & 1.79  & 0.14 \\
			\hline
		\end{tabular}
	\end{table}

	\begin{table}[!t]
		\centering
		\caption{Coverage rates for DC LEQR with iterations $q=1,2,3,4,5$ for different choices of the scaling constant $c$. Noises $\epsilon_i$'s are generated from heteroscedastic normal distribution. Dimension $p=15$. Batch size $m=100$. Sample size $n$ varies from $m^{1.6}$ to $m^{3}$. Quantile level $\tau=0.1$.}
		\label{table:adaptive_2}
		\begin{tabular}{lc|ccccc}
			\hline
			& $\log_m(n)$ & $q=1$ & $q=2$ & $q=3$ & $q=4$ &$q=5$ \\
			\hline
			$c=1$    &     &          &          &  &        &          \\
			& 1.6 & 0.477 & 0.851 & 0.902 & 0.941 &0.940 \\
			& 2.0 & 0.157 & 0.896 & 0.949 & 0.948 &0.948\\
			& 2.4 & 0.014 & 0.614 & 0.715 & 0.947 &0.947\\
			& 3.0 & 0.000 & 0.112 & 0.431 & 0.952 &0.950\\
			\hline
			$c=3$    &     &          &          & &         &          \\
			& 1.6 & 0.562 & 0.910 & 0.939 & 0.944 &0.946\\
			& 2.0 & 0.249 & 0.932 & 0.951 & 0.952 &0.949\\
			& 2.4 & 0.003 & 0.681 & 0.751 & 0.947 &0.947\\
			& 3.0 & 0.000 & 0.171 & 0.529 & 0.934 &0.942\\
			\hline
			$c=5$    &     &          &          &          &          \\
			& 1.6 & 0.421 & 0.754 & 0.882 & 0.943 &0.946\\
			& 2.0 & 0.134 & 0.648 & 0.845 & 0.946 &0.948\\
			& 2.4 & 0.006 & 0.311 & 0.676 & 0.936 &0.944\\
			& 3.0 & 0.000 & 0.101 & 0.359 & 0.907 &0.941\\
			\hline
			$c=10$   &     &          &          &          &          \\
			& 1.6 & 0.348 & 0.719 & 0.842 & 0.946 &0.951\\
			& 2.0 & 0.106 & 0.684 & 0.820 & 0.955 &0.952\\
			& 2.4 & 0.000 & 0.218 & 0.421 & 0.954 &0.950\\
			& 3.0 & 0.000 & 0.008 & 0.328 & 0.884 &0.919\\
			\hline
			data-adaptive &     &          &          &          &          \\
			& 1.6 & 0.615 & 0.945 & 0.948 & 0.949 &0.949\\
			& 2.0 & 0.357 & 0.913 & 0.951 & 0.952 &0.953\\
			& 2.4 & 0.049 & 0.649 & 0.794 & 0.949 &0.951\\
			& 3.0 & 0.010 & 0.217 & 0.534 & 0.935 &0.946\\
			\hline
		\end{tabular}
	\end{table}
	
	\begin{table}[!t]
		\centering
		\caption{Coverage rates for DC LEQR with iterations $q=1,2,3,4,5$ for different choices of the scaling constant $c$. Noises $\epsilon_i$'s are generated from exponential distribution. Dimension $p=15$. Batch size $m=100$. Sample size $n$ varies from $m^{1.6}$ to $m^{3}$. Quantile level $\tau=0.1$.}
		\label{table:adaptive_3}
		\begin{tabular}{lc|ccccc}
			\hline
			& $\log_m(n)$ & $q=1$ & $q=2$ & $q=3$ & $q=4$ & $q=5$ \\
			\hline
			$c=1$    &     &          &          &          &        &  \\
			& 1.6 & 0.313 & 0.745 & 0.832 & 0.880&0.916 \\
			& 2.0 & 0.116 & 0.678 & 0.846 & 0.920&0.942 \\
			& 2.4 & 0.014 & 0.514 & 0.726 & 0.915&0.931 \\
			& 3.0 & 0.000 & 0.127 & 0.429 & 0.907&0.931 \\
			\hline
			$c=3$    &     &          &          &          &          \\
			& 1.6 & 0.363 & 0.754 & 0.864 & 0.921 &0.944\\
			& 2.0 & 0.145 & 0.715 & 0.869 & 0.907&0.929 \\
			& 2.4 & 0.079 & 0.528 & 0.758 & 0.912&0.943 \\
			& 3.0 & 0.000 & 0.184 & 0.541 & 0.938 &0.946\\
			\hline
			$c=5$    &     &          &          &          &          \\
			& 1.6 & 0.296 & 0.749 & 0.898 & 0.871 &0.922\\
			& 2.0 & 0.101 & 0.654 & 0.865 & 0.904&0.936\\
			& 2.4 & 0.012 & 0.544 & 0.731 & 0.891&0.914 \\
			& 3.0 & 0.000 & 0.131 & 0.467 & 0.842&0.876 \\
			\hline
			$c=10$   &     &          &          &          &          \\
			& 1.6 & 0.216 & 0.698 & 0.812 & 0.896 &0.917\\
			& 2.0 & 0.031 & 0.610 & 0.764 & 0.889 &0.914\\
			& 2.4 & 0.002 & 0.419 & 0.713 & 0.876 &0.897\\
			& 3.0 & 0.000 & 0.006 & 0.308 & 0.812 &0.844\\
			\hline
			data-adaptive &     &          &          &          &          \\
			& 1.6 & 0.377 & 0.798& 0.899 & 0.917 &0.949\\
			& 2.0 & 0.209 & 0.745 & 0.894 & 0.902 &0.941\\
			& 2.4 & 0.121 & 0.605 & 0.814 & 0.924 &0.936\\
			& 3.0 & 0.000 & 0.218 & 0.699 & 0.900 &0.913\\
			\hline
		\end{tabular}
	\end{table}
	\begin{table}[!t]
		\centering
		\caption{Square root of the ratio of the estimated variance and the true limiting variance  using  different choices of the scaling constant $c$ in bandwidths. Noises $\epsilon_i$'s are generated from heteroscedastic normal distribution. Dimension $p=15$. Batch size $m=100$. Sample size $n$ varies from $m^{1.6}$ to $m^{3}$. Quantile level $\tau=0.1$.}
		\label{table:sesd_adaptive_2}
		\begin{tabular}{lc|ccccc}
			\hline
			& $\log_m(n)$ & $q=1$ & $q=2$ & $q=3$ & $q=4$ &$q=5$ \\
			\hline
			$c=1$    &     &          &          &          &          \\
			& 1.6 & 1.15 & 1.10 & 1.08 & 1.06 &1.02\\
			& 2.0 & 1.42 & 1.19 & 1.12 & 1.08 &1.03\\
			& 2.4 & 1.23 & 1.14 & 1.06 & 1.04 &1.01\\
			& 3.0 & 1.06 & 1.04 & 0.99 & 0.98 &0.97\\
			\hline
			$c=3$    &     &          &          &          &          \\
			& 1.6 & 1.10 & 1.05 & 1.02 & 1.03 &1.01\\
			& 2.0 & 1.31 & 0.11 & 1.07 & 1.02 &1.00\\
			& 2.4 & 1.17 & 1.04 & 1.03 & 1.03 &1.03\\
			& 3.0 & 1.04 & 1.01 & 1.00 & 1.00  &0.99\\
			\hline
			$c=5$    &     &          &          &          &          \\
			& 1.6 & 1.12 & 1.01 & 1.03 & 1.03 &1.03\\
			& 2.0 & 1.37 & 1.21 & 1.10 & 1.07 &1.01\\
			& 2.4 & 1.22 & 1.15 & 1.05 & 1.04 &1.04\\
			& 3.0 & 1.14 & 1.09 & 1.02 & 1.02 &1.01\\
			\hline
			$c=10$   &     &          &          &          &          \\
			& 1.6 & 1.19 & 1.09 & 0.98 & 0.96 &0.94\\
			& 2.0 & 1.47 & 1.17 & 1.19 & 1.14 &1.10\\
			& 2.4 & 1.25 & 1.20 & 1.11 & 1.12 &1.07\\
			& 3.0 & 1.11 & 1.04 & 1.04 & 1.03 &1.03\\
			\hline
			data-adaptive &     &          &          &          &          \\
			& 1.6 & 1.04 & 1.07 & 0.99 & 1.03 &1.05\\
			& 2.0 & 1.06 & 1.05 & 1.03 & 1.02 &1.02\\
			& 2.4 & 1.07 & 1.04 & 1.00 & 1.01 &1.03\\
			& 3.0 & 1.06 & 1.01 & 0.94 & 0.97 &0.98\\
			\hline
		\end{tabular}
	\end{table}
	\begin{table}[!t]
		\centering
		\caption{Square root of the ratio of the estimated variance and the true limiting variance  using  different choices of the scaling constant $c$ in bandwidths. Noises $\epsilon_i$'s are generated from exponential distribution. Dimension $p=15$. Batch size $m=100$. Sample size $n$ varies from $m^{1.6}$ to $m^{3}$. Quantile level $\tau=0.1$.}
		\label{table:sesd_adaptive_3}
		\begin{tabular}{lc|ccccc}
			\hline
			& $\log_m(n)$ & $q=1$ & $q=2$ & $q=3$ & $q=4$&$q=5$ \\
			\hline
			$c=1$    &     &          &          &          &   &       \\
			& 1.6 & 1.29 & 1.18 & 1.06 & 1.04 &1.04\\
			& 2.0 & 1.16 & 1.08 & 1.02 & 1.00 &0.01\\
			& 2.4 & 1.14 & 1.12 & 1.07 & 1.06 &1.04\\
			& 3.0 & 1.06 & 1.04 & 1.02 & 1.02 &1.02\\
			\hline
			$c=3$    &     &          &          &  &        &          \\
			& 1.6 & 1.17 & 1.13 & 1.01 & 0.97 &0.98\\
			& 2.0 & 1.08 & 1.04 & 1.00 & 0.99 &1.00\\
			& 2.4 & 1.04 & 0.99 & 0.99 & 0.99 &1.00\\
			& 3.0 & 1.02 & 0.94 & 0.97 & 1.01 &1.03\\
			\hline
			$c=5$    &     &          &          &          &          \\
			& 1.6 & 1.24 & 1.11 & 1.09 & 1.11&1.09 \\
			& 2.0 & 1.15 & 1.01 & 1.08 & 1.07 &1.07\\
			& 2.4 & 1.06 & 1.04 & 1.04 & 1.01 &1.01\\
			& 3.0 & 1.03 & 1.01 & 1.02 & 1.01 &1.01\\
			\hline
			$c=10$   &     &          &          &          &          \\
			& 1.6 & 1.44 & 1.26 & 1.10 & 1.07 &1.07\\
			& 2.0 & 1.31 & 0.77 & 0.96 & 0.97 &0.98\\
			& 2.4 & 1.08 & 0.89 & 0.96 & 0.99 &1.00\\
			& 3.0 & 1.17 & 1.09 & 1.07 & 1.05 &1.04\\
			\hline
			data-adaptive &     &          &          &          &          \\
			& 1.6 & 1.12 & 1.04 & 1.03 & 1.00 &1.01\\
			& 2.0 & 1.04 & 0.96 & 0.97 & 1.01&1.02 \\
			& 2.4 & 1.03 & 0.99 & 1.00 & 1.00 &1.01\\
			& 3.0 & 1.01 & 0.96 & 0.99 & 1.00 &0.98\\
			\hline
		\end{tabular}
	\end{table}

	\begin{table}[!t]
		\centering
		\caption{Coverage rates, bias, and variance of DC LEQR v.s. QR All (i.e., the QR estimator in \eqref{eq:QR} computed on all data points) and na\"ive-DC when dimension $p=1000$. Noises $\epsilon_i$'s are generated from homoscedastic normal distribution.  The batch size $m \in \{10,000,20,000\}$ and the sample size $n \in \{10^6, 1.5\times 10^6, 2\times 10^6\}$.  Note that the method QR All does not depend on the batch size $m$ and thus the reported numbers are the same for $m=10,000$ and $20,000$. The quantile level is $\tau = 0.1$, and the nominal level is $95\%$.}
		\vspace{0.3cm}
		\label{table:large_1}
		\setlength{\tabcolsep}{3pt}
		\begin{tabular}{ll|ccc|ccc}
			\hline
			& $n$    & Coverage & Bias& Variance & Coverage   & Bias& Variance \\
			&   & Rate& $(\times 10^{-2})$ & $(\times 10^{-4})$ & Rate  & $(\times 10^{-2})$ & $(\times 10^{-4})$ \\
			\hline
			&   & \multicolumn{3}{c|}{DC LEQR $q=4$} &\multicolumn{3}{c}{DC LEQR $q=8$}\\
			\hline
			\multicolumn{2}{l|}{$m=10,000$\qquad~}   &&&&  &&\\
			\qquad~& $n=10^6$           & 0.950  & ~0.21  & 68.65 & 0.955    & ~0.52  & 37.95 \\
			& $n=1.5\times 10^6$ & 0.964  & -0.80 & 32.50 & 0.958    & -0.27 & 18.17 \\
			& $n=2\times 10^6$   & 0.967  & -0.22 & 12.59 & 0.955    & -0.34 & 11.73 \\
			\multicolumn{2}{l|}{$m=20,000$\qquad~}   &&&&  &&\\
			& $n=10^6$           & 0.952  & ~0.08  & 59.57 & 0.953    & -0.25 & 25.34 \\
			& $n=1.5\times 10^6$ & 0.961  & -0.04 & 18.66 & 0.951    & ~0.13  & 17.70 \\
			& $n=2\times 10^6$   & 0.964  & -0.25 & 11.86 & 0.956    & -0.15 & 12.46 \\
			\hline
			&   & \multicolumn{3}{c|}{QR All} &\multicolumn{3}{c}{Na\"ive-DC} \\
			\hline
			\multicolumn{2}{l|}{$m=10,000$\qquad~}   &&&&  &&\\
			& $n=10^6$           & 0.948  & ~0.59  & 21.66 & 0.797    & ~5.13  & 36.90 \\
			& $n=1.5\times 10^6$ & 0.952  & -0.16 & 18.39 & 0.739    & ~5.40  & 17.42 \\
			& $n=2\times 10^6$   & 0.946  & -0.11 & \,\;9.36  & 0.768    & ~4.48  & 14.98 \\
			\multicolumn{2}{l|}{$m=20,000$\qquad~}   &&&&  &&\\
			& $n=10^6$           & 0.948  & ~0.59  & 21.66  & 0.957    & ~2.17  & 22.91 \\
			& $n=1.5\times 10^6$ & 0.952  & -0.16 & 18.39  & 0.913    & ~1.83  & 18.84 \\
			& $n=2\times 10^6$   & 0.946  & -0.11 &  \,\;9.36   & 0.899    & ~1.94  & 14.58\\
			\hline
		\end{tabular}
		
	\end{table}
	\begin{table}[!t]
		\centering
		\caption{Coverage rates, bias, and variance of DC LEQR v.s. QR All (i.e., the QR estimator in \eqref{eq:QR} computed on all data points) and na\"ive-DC when dimension $p=1000$. Noises $\epsilon_i$'s are generated from heteroscedastic normal distribution. The batch size $m \in \{10,000,20,000\}$ and the sample size $n \in \{10^6, 1.5\times 10^6, 2\times 10^6\}$.  Note that the method QR All does not depend on the batch size $m$ and thus the reported numbers are the same for $m=10,000$ and $20,000$. The quantile level is $\tau = 0.1$, and the nominal level is $95\%$.}
		\vspace{0.3cm}
		\label{table:large_2}
		\setlength{\tabcolsep}{3pt}
		\begin{tabular}{ll|ccc|ccc}
			\hline
			& $n$    & Coverage & Bias& Variance & Coverage   & Bias& Variance \\
			&   & Rate& $(\times 10^{-2})$ & $(\times 10^{-4})$ & Rate  & $(\times 10^{-2})$ & $(\times 10^{-4})$ \\
			\hline
			&   & \multicolumn{3}{c|}{DC LEQR $q=4$} &\multicolumn{3}{c}{DC LEQR $q=8$}\\
			\hline
			\multicolumn{2}{l|}{$m=10,000$\qquad~}   &&&&  &&\\
			\qquad~& $n=10^6$ & 0.895    & ~0.88& 61.30   & 0.928 & -0.22    & 38.55   \\
			& $n=1.5\times 10^6$ & 0.964    & ~0.41& 38.07    & 0.957 & -0.51    & 23.24    \\
			& $n=2\times 10^6$   & 0.961    & -0.47    & 17.40    & 0.954 & -0.58    & 15.61    \\
			\multicolumn{2}{l|}{$m=20,000$\qquad~}   &&&&  &&\\
			& $n=10^6$ & 0.920    & -1.21    & 54.29   & 0.926 & ~0.95& 35.07   \\
			& $n=1.5\times 10^6$ & 0.960    & ~0.08& 32.54    & 0.953 & ~0.08& 21.67    \\
			& $n=2\times 10^6$   & 0.966    & -0.43    & 19.38    & 0.955 & -0.48    & 18.20    \\
			\hline
			&   & \multicolumn{3}{c|}{QR All} &\multicolumn{3}{c}{Na\"ive-DC} \\
			\hline
			\multicolumn{2}{l|}{$m=10,000$\qquad~}   &&&&  &&\\
			& $n=10^6$ & 0.954    & ~0.61& 31.14    & 0.959 & ~3.44& 37.35    \\
			& $n=1.5\times 10^6$ & 0.954    & -0.34    & 21.80    & 0.945 & ~3.41& 21.72    \\
			& $n=2\times 10^6$   & 0.947    & ~0.31& 14.71    & 0.808 & ~3.30& 22.53    \\
			\multicolumn{2}{l|}{$m=20,000$\qquad~}   &&&&  &&\\
			& $n=10^6$ &0.954    & ~0.61& 31.14 & 0.952 & ~1.21& 32.90    \\
			& $n=1.5\times 10^6$ & 0.954    & -0.34    & 21.80 & 0.945 & ~1.26& 26.15    \\
			& $n=2\times 10^6$   &0.947    & ~0.31& 14.71  & 0.904 & ~0.50& 18.97  \\
			\hline
		\end{tabular}
	\end{table}
	
	\begin{table}[!t]
		\centering
		\caption{Coverage rates (nominal level $95\%$), bias ($\times 10^{-2})$, variance ($\times 10^{-4})$,  and computation time (in seconds) of DC LEQR for different $q$ versus the ADMM-based QR.  Noises $\epsilon_i$'s are generated from homoscedastic normal distribution. The quantile level is $\tau = 0.1$. Dimension $p=15$.}
		\label{table:ADMM}
		\begin{tabular}{llrrrrr}
			\hline
			&& \multicolumn{4}{c}{DC LEQR}   & ADMM\\
			&&$q=1$&$q=2$   &$q=3$  & $q=4$   &  QR  \\
			\hline
			
			$m=100, n=10^6$  &&&    &    &    &\\
			\quad\quad\quad  &   Bias &6.112&0.865&0.038&-0.020&-0.021\\
			& Variance &35.464&1.637&0.037&0.027&0.029\\
			& Coverage &0.001&0.314&0.950&0.951&0.949\\
			& Time& 0.379    & 0.632   & 0.861   & 1.102  & 1.596 \\
			\hline
			$m=500, n=10^6$  &&&    &    &    &\\
			& Bias&1.334&0.029&-0.008&-0.010&-0.011\\
			& Variance &0.642&0.042&0.029&0.029&0.027\\
			& Coverage &0.087&0.914&0.947&0.951&0.951\\
			& Time& 1.216    & 2.354   & 3.499   & 4.618  & 5.542 \\
			\hline
			$m=500, n=10^7$  &&&    &    &    &\\
			& Bias&1.171&0.046&-0.005&-0.005&-0.005\\
			& Variance &0.885&0.016&0.002&0.002&0.002\\
			& Coverage &0.024&0.642&0.943&0.948&0.946\\
			& Time& 1.429   & 2.604  & 3.767 & 4.891 & 5.779 \\
			\hline
			$m=1000, n=10^7$ &&&    &    &    &\\
			& Bias&0.786&0.016&-0.007&-0.007&-0.007\\
			& Variance &0.167&0.002&0.003&0.003&0.003\\
			& Coverage &0.014&0.897&0.952&0.947&0.948\\
			& Time& 3.454   & 5.233  & 7.119 & 9.344 & 11.423 \\
			\hline
		\end{tabular}
	\end{table}

	\begin{figure}[!t]
		\centering
		\subfigure[$p=3$, $\epsilon_i\sim N(0,1)$]{\includegraphics[width = 0.48\textwidth]{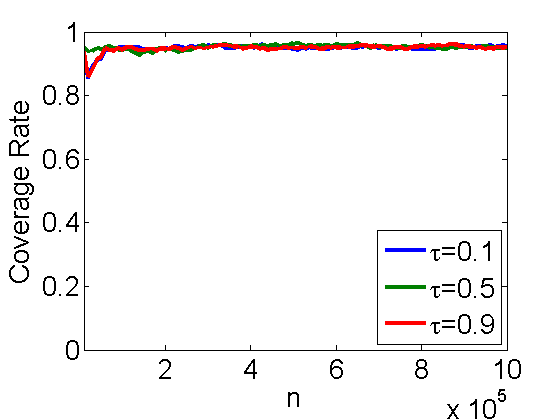}\label{fig:online_1}}
		\subfigure[$p=3$, $\epsilon_i\sim \exp(1)$]{\includegraphics[width = 0.48\textwidth]{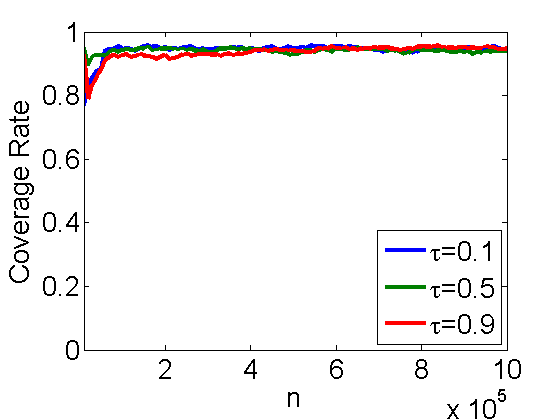}\label{fig:online_2}}
		\subfigure[$p=15$, $\epsilon_i\sim N(0,1)$]{\includegraphics[width = 0.48\textwidth]{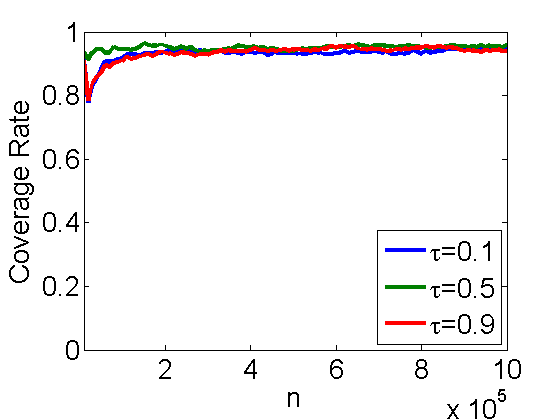}\label{fig:online_3}}
		\subfigure[$p=15$, $\epsilon_i\sim \exp(1)$]{\includegraphics[width = 0.48\textwidth]{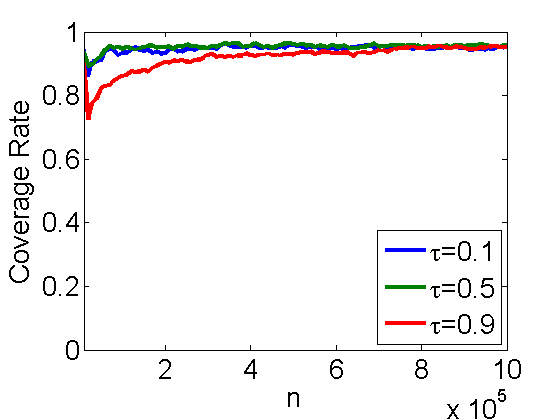}\label{fig:online_4}}
		\caption{Coverage rates of online LEQR when $m=500$.}
		\label{fig:online}
	\end{figure}

\end{document}